\documentclass[11pt]{article}

\usepackage[stable]{footmisc}
\usepackage{amsfonts}
\usepackage{amssymb,amsmath}
\usepackage{amsthm,amsgen}
\usepackage{graphicx,epsfig}
\usepackage[mathscr]{eucal}
\usepackage{slashed}
\usepackage{bbm}
\usepackage{ marvosym }
\usepackage{dsfont}
\usepackage{accents}
\usepackage{color}

\usepackage[
backref,pagebackref]{hyperref}

\newtheorem{thm}{THEOREM}[section]
\newtheorem*{thm*}{THEOREM}

\newtheorem{lem}[thm]{Lemma}
\newtheorem{prop}[thm]{Proposition}
\theoremstyle{definition}

\newtheorem{rem}[thm]{Remark}

\newcommand{\drm}{\mathrm{d}}

\newcommand{\pddt}{\frac{\partial\phantom{t}}{\partial t}}

\DeclareMathOperator{\tr}{tr}

\newcommand{\refeq}[1]{(\ref{#1})}

\newcommand{\vect}[1] {\boldsymbol{{ #1}} }

\newcommand{\tenseur}[1]{{\textbf{\textsf{#1}}}}

\newcommand{\bMj}{{\tenseur{j}}}

\newcommand{\bpi}{\boldsymbol{\pi}}

\newcommand{\jV}{{\vect{j}}}		
\newcommand{\kV}{{\vect{k}}}		
\newcommand{\pV}{{\vect{p}}}            
\newcommand{\qV}{{\vect{q}}}            
\newcommand{\sV}{{\vect{s}}}            
\newcommand{\vV}{{\vect{v}}}            
\newcommand{\AV}{\pmb{\mathcal{A}}}

\newcommand{\FV}{\pmb{\mathcal{F}}}

\numberwithin{equation}{section}

\newcounter{subequation}
	\newenvironment{subequation}%
	{\addtocounter{equation}{-1}%
	\stepcounter{subequation}%
	\begin{equation}}%
	{\end{equation}%
}

\newcommand{\fa}{\mathfrak{a}}

\newcommand{\fe}{\mathfrak{e}}
\newcommand{\ff}{\mathfrak{f}}

\newcommand{\fm}{\mathfrak{m}}

\newcommand{\bB}{\mathbf{B}}
\newcommand{\bb}{\mathbf{b}}
\newcommand{\ba}{\mathbf{a}}
\newcommand{\bfa}{\boldsymbol{\fa}}

\newcommand{\bE}{\mathbf{E}}
\newcommand{\be}{\mathbf{e}}
\newcommand{\bF}{\mathbf{F}}
\newcommand{\beff}{\mathbf{f}}
\newcommand{\bff}{\boldsymbol{\ff}}

\newcommand{\bj}{\mathbf{j}}

\newcommand{\bL}{\mathbf{L}}

\newcommand{\bn}{\mathbf{n}}

\newcommand{\bP}{\mathbf{P}}

\newcommand{\bq}{\mathbf{q}}
\newcommand{\bR}{\mathbf{R}}

\newcommand{\bt}{\mathbf{t}}

\newcommand{\bu}{\mathbf{u}}

\newcommand{\bX}{\mathbf{X}}
\newcommand{\bx}{\mathbf{x}}

\newcommand{\by}{\mathbf{y}}
\newcommand{\bz}{\mathbf{z}}

\newcommand{\btau}{\boldsymbol{\tau}}

\newcommand{\beq}{\begin{equation}}
\newcommand{\eeq}{\end{equation}}
\newcommand{\bseq}{\begin{subequation}}
\newcommand{\eseq}{\end{subequation}}

\newcommand{\Id}{\mathds{1}}

\newcommand{\fL}{\mathfrak{L}}

\newcommand{\p}{\partial}

\newcommand{\cA}{\mathcal{A}}
\newcommand{\cB}{\mathcal{B}}

\newcommand{\cE}{{\mathcal E}}

\newcommand{\cH}{{\mathcal H}}

\newcommand{\cL}{{\mathcal L}}
\newcommand{\cLph}{{\mathcal L}_{\mbox{\tiny{ph}}}}

\newcommand{\cM}{{\mathcal M}}

\newcommand{\cP}{{\mathcal P}}

\newcommand{\cS}{{\mathcal S}}
\newcommand{\cT}{\mathcal{T}}

\newcommand{\cV}{\mathcal{V}}

\newcommand{\siV}{\boldsymbol{\sigma}}
\newcommand{\xiV}{{\boldsymbol{\xi}}}

\newcommand{\bLa}{\boldsymbol{\Lambda}}

\newcommand{\psiPH}{{\psi_{\mbox{\rm\tiny ph}}}}
\newcommand{\psiPHb}{{\overline\psi_{\mbox{\rm\tiny ph}}}}

\newcommand{\psiEL}{{\psi_{\mbox{\rm\tiny el}}}}
\newcommand{\psiELb}{{\overline\psi_{\mbox{\rm\tiny el}}}}

\newcommand{\psiLP}{{{\boldsymbol{\psi}}_{\mbox{\rm\tiny LP}}}}

\newcommand{\PsiPh}{{}^{\mbox{\textrm{`}}}{\Psi_{\mbox{\rm\tiny ph}}^{\mbox{\textrm{'}}}}}

\newcommand{\phiPH}{{\phi_{\mbox{\rm\tiny ph}}}}

\newcommand{\mEL}{{m_{\mbox{\rm\tiny el}}}}
\newcommand{\mPHO}{{m_{\mbox{\rm\tiny ph}}}}
\newcommand{\mPH}{{m_{\mbox{\Lightning}}}}

\newcommand{\Cset}{{\mathbb C}}

\newcommand{\Nset}{{\mathbb N}}

\newcommand{\Rset}{{\mathbb R}}

\newcommand{\Zset}{{\mathbb Z}}

\newcommand{\pdt}{{\partial_t^{\phantom{0}}}}

\newcommand{\ze}{\zeta}

\newcommand{\la}{\lambda}

\newcommand{\de}{\delta}

\newcommand{\al}{\alpha}

\newcommand{\ga}{\gamma}

\newcommand{\bga}{\boldsymbol{\gamma}}
\newcommand{\ep}{\epsilon}

\newcommand{\si}{\sigma}
\newcommand{\Si}{\Sigma}
\newcommand{\nab}{\nabla}
\newcommand{\half}{\frac{1}{2}}

\newcommand{\diag}{\mbox{diag}}
\newcommand{\Span}{\mbox{Span}}

\newcommand{\bna}{\begin{eqnarray}}
\newcommand{\ena}{\end{eqnarray}}
\newcommand{\bea}{\begin{eqnarray}}
\newcommand{\eea}{\end{eqnarray}}
\newcommand{\ben}{\begin{enumerate}}
\newcommand{\een}{\end{enumerate}}
\newcommand{\bi}{\begin{itemize}}
\newcommand{\ei}{\end{itemize}}

\newcommand{\dal}{\raisebox{3pt}{\fbox{}}\,}

\newcommand{\sdf}{{}^{\mathrm{sd}}\bff}
\newcommand{\asdf}{{}^{\mathrm{asd}}\bff}

\newcommand{\TMmunu}{{T_{\mbox{\tiny\textrm{M}}}^{\mu\nu}}}

\begin{document}

\title{On the Quantum-Mechanics of a Single Photon}

\author{\normalsize\sc{Michael K.-H. Kiessling and A. Shadi Tahvildar-Zadeh}\\
	{$\phantom{nix}$}\\[-0.1cm] 
        \normalsize Department of Mathematics\\[-0.1cm]
	Rutgers, The State University of New Jersey\\[-0.1cm]
	110 Frelinghuysen Rd., Piscataway, NJ 08854}
\vspace{-0.3cm}
\date{Submitted to J. Math. Phys. on 12/31/2017\\
 most recent revision: 10/09/2018} 
\maketitle
%
%
%
\begin{abstract}
\noindent
 It is shown that a Dirac(-type) equation for a rank-two bi-spinor field $\psiPH$ on Minkowski (configuration) spacetime
furnishes a Lorentz-covariant \emph{quantum-mechanical} wave equation in position-space representation for a single free photon.
 This equation does not encounter any of the roadblocks that have obstructed previous attempts (by various authors) to formulate 
a {quantum-mechanical} photon wave equation.  
 In particular, it implies that the photon wave function $\psiPH$ yields conserved 
non-negative Born-rule-type quantum probabilities,  
and that its probability current density four-vector transforms properly under Lorentz transformations.
 Moreover, the eigenvalues of the pertinent photon Dirac Hamiltonian and the vector eigenvalues of the photon momentum operator 
yield the familiar Einstein relations $E=\hbar\omega$ and $\pV=\hbar\kV$, respectively. 
 Furthermore, these spin-1 wave modes are automatically transversal without the need of an additional constraint on the initial data.
 Some comments on  other proposals to set up a photon wave equation are supplied as well.
\end{abstract}

\smallskip

\vfill
\hrule
\smallskip
\copyright{2018. The authors. Reproduction for non-commercial purposes only is permitted.}

\medskip

\newpage

\section{Introduction} 

 According to Weinberg (\cite{WeinbergBOOKqft}, p.3), during the last century quantum-field theorists had come 
to the conclusion that ``relativistic wave mechanics, in the sense of a relativistic quantum theory of a fixed 
number of particles, is an impossibility.''
 However, no theorem exists to this effect, and it is hard to imagine that a theorem of this kind, had any been claimed, 
could deliver such a sweeping verdict.
 In contrast to readily believable statements of the kind that 
``it has been proved, according to a certain theory, that  a certain phenomenon is impossible,'' 
claims of the impossibility of a certain kind of theory should always be received with a healthy dose of skepticism.\footnote{Any 
  theorem to this effect would involve certain mathematical assumptions --- which 
  might very well be too restrictive to be physically significant.
     Take for example John von Neumann's assertion \cite{JvN} of \emph{the impossibility of hidden-variable theories that would reproduce 
  the predictions of quantum mechanics}. It fell flat on its face by ignoring the theory of de Broglie \cite{deBroglieSOLVAY} 
(which was rediscovered and further developed by Bohm \cite{Bohm52}.) 
See \cite{BricmontBOOK}, in particular the John Bell quote on p. 257.}
         
 As mathematical physicists we want to find out how much of empirical electromagnetism can be accurately accounted for in terms of 
a relativistic $N$-body quantum theory of electrons, photons, and their anti-particles.
 Inspired by Bell's writings \cite{BellBOOKsec}, we in particular want to find out whether it is possible to formulate
such a theory as a generalization of the non-relativistic theory of de Broglie \cite{deBroglieSOLVAY} and Bohm \cite{Bohm52}
in which the quantum-mechanical wave function $\psi$ guides the actual motion of these particles.
 The non-relativistic theory\footnote{See also \cite{BoHi} and \cite{HollandBOOK}.}
 was worked out in mathematical and conceptual detail by D\"urr, Goldstein, Zangh\`{\i} and their school \cite{DGZ,DT}, 
 who also made serious advances toward its relativistic generalization \cite{BerndlETal,DGNSZ}.
 Ultimately such a theory, if possible at all, would of course have to be generalized to account for the other forces and 
particles of nature; yet, one step at a time. 

 In this paper we concern ourselves mostly with the problem of formulating relativistic quantum-mechanical wave equations for
the fundamental constitutive particles of electromagnetism.
 As is well-known, this is non-trivial already when $N=1$, and truly challenging for $N=2$ interacting particles. 
 On the bright side,  once $N=2$ interacting particles (one electron and one photon) can be handled, the formulation of the 
general $N$-body model should be a straightforward task, assuming $N$-body interactions decompose into a combination of pairwise 
electron-photon interactions.

 There are two $N=1$ sub-problems: the single electron problem and the single photon problem.
 The single electron problem is treated extensively in the literature on quantum mechanics,
whether introductory (e.g. \cite{Dirac1930}, \cite{LandauLifshitzBOOKrelQM}, \cite{BohmBOOKonQM}, \cite{MessiahQMbook}) 
or advanced monographs (e.g. \cite{GreinerETal}, \cite{ThallerBOOK}, \cite{WeinbergBOOKqft}), 
where Dirac's equation \cite{Dirac1928a,Dirac1928b} 
\beq \label{eq:DiracEQelectron} 
\boxed{\ga^\mu p_\mu \psiEL + \mEL c \Id \psiEL =0}
\eeq
is presented as the inevitable Lorentz-covariant \emph{quantum-mechanical} wave equation for a single free electron,
and its generalization via so-called ``minimal coupling'' $p_\mu\mapsto p_\mu + \frac{e}{c}A_\mu$ features as the Dirac wave 
equation for a single electron in ``externally generated, classical electromagnetic fields.'' 
 Yet the precise quantum-mechanical meaning of the negative energy continuum below $-\mEL c^2$ 
has been debated ever since the equation's inception \cite{Dirac1928a,Dirac1928b},\footnote{Recall 
  that, to make sense of his equation, Dirac found it necessary to postulate the existence of uncountably many unobserved
 electrons, occupying the ``negative energy states,'' and to invoke Pauli's exclusion principle.}
and that surely has contributed to the widespread perception amongst quantum physicists expressed above in Weinberg's quote\footnote{Weinberg's assessment is not universally accepted. After all, it cannot be denied that Dirac's wave equation 
for a single electron \refeq{eq:DiracEQelectron} \emph{exists} and has yielded results which compare 
favorably with experimental facts without having recourse to Dirac's ``filled sea of negative energy states.''
 For instance, St\"uckelberg \cite{Stueckelberg} and Feynman \cite{FeynmanD} have proposed that Dirac's wave equation is 
\emph{in effect a quantum-mechanical two-particle equation}, since it seems to account simultaneously for the electron and its 
anti-particle, the positron; see \cite{Thaller} and Thaller's book \cite{ThallerBOOK} for a careful discussion of these ideas. 
 Yet this two-particle interpretation has its own problems. 
 In particular, according to the usual formalism for quantum-mechanical wave functions on $N$-particle configuration space,
Dirac's equation ``for the electron'' simply \emph{is} a one-particle equation, for Dirac's bi-spinor field lives on a 
  \emph{single-}, not a two-particle configuration space(time). 
 A novel interpretation of Dirac's equation as a single-particle equation, in which electron and anti-electron are merely
``two sides of the same coin,'' has been offered in \cite{KTZzGKNDiia}.}.

 By contrast, the single photon problem has fared far worse.
 Even the mere formulation of a one-photon counterpart to Dirac's relativistic quantum-mechanical wave equation for one electron
has frequently been declared either an impossible task, or else a trivial one.
 Both views are still widespread, but clearly they cannot simultaneously be both correct!
  
 Regarding the perceived triviality of this task,  one often finds the claim (mostly in the quantum optics literature) that the electromagnetic 
Maxwell field is the photon wave function, and Maxwell's field equations are the photon wave equation --- in disguise!  
This point of view seems to have its origin in Wigner's seminal work \cite{Wig1939} on the unitary representations of the Poincar\'e group,
 where he seeks to derive relativistic wave equations for particles of any mass and spin, based purely on representation-theoretic 
considerations;\footnote{Incidentally, it should be noted that such a purely algebraic derivation of the 
relativistic wave equations for a particle is not capable of telling us anything about the object that satisfies those equations, beyond 
how that object should transform under Poincar\'e transformations.\vspace{-10pt}}
 see also \cite{BarWig1948}.
  There he writes that for massless particles and ``for $s = 0$ we have simply the equation $\dal \varphi = 0$,
for $s = \pm\frac{1}{2}$ Dirac's electron equation without mass, for $s = \pm 1$ Maxwell's electromagnetic equations, etc." 
 Wigner does not write down these ``Maxwell's [...] equations,'' but in \cite{Wei1964} Weinberg repeats Wigner's assertion, 
and adds ``...they are just Maxwell's free-space equations for left- and right-circularly polarized radiation:
\bna
\nabla \times [\bE - i\bB] + i(\p/\p t)[\bE - i\bB] & = & 0\\
\nabla \times [\bE + i \bB] - i(\p/\p t)[\bE + i\bB] & = & 0,\mbox{"}
\ena
thus making it clear that 
    only the evolutionary subset of Maxwell's equations is obtained.
 More importantly, note that this system of evolutionary equations is redundant, because
the complex fields $\bE - i\bB=:\boldsymbol{\Psi}_-$ and $\bE + i\bB=:\boldsymbol{\Psi}_+$ 
are complex conjugates of each other.
 Yet to ensure invariance under space reflections, both equations are needed to form such an invariant system. 
 While Weinberg, as well as Bia{\l}ynicki-Birula \cite{IBBphotonREV}, take the two complex fields 
to be conjugates of one another, other authors, including Oppenheimer \cite{OppiPHOTON} and Penrose \cite{Pen1976}, 
have compellingly argued that the two fields $\boldsymbol{\Psi}_-$ and $\boldsymbol{\Psi}_+$ 
could not be conjugates, but in fact need to be two independent fields. 
 We will come back to this issue in Appendix \ref{app:others}.

 It is instructive to note that during the ensuing 30+ years, Weinberg switched from one view to the other, for
in \cite{WeinbergTALK}, p.2., we find him state: ``Certainly the Maxwell field is not the wave function of the photon...''.
 We are not aware of when and why exactly he changed his opinion.

 Other authors have expressed serious doubts about the possibility of formulating a relativistic wave equation for a photon wave function.
 The obstructions to setting up such a Lorentz-covariant ``theory of light quanta'' were summarized already by 
Oppenheimer, who identified \vspace{-5pt}
\begin{quotation}
\noindent
``a very grave and inescapable defect of the theory. The [quantum-mechanical wave] equations themselves are co-variant; 
but the Lagrangian from which they and their complex conjugate equations may be deduced is not a [Lorentz] scalar quantity. 
 This has important and disastrous consequences, for it means that the energy density of the quanta is not the [time-time]
component of a second rank tensor, nor the momentum density the [time-space] component of such a tensor.
 Further, and equally disastrous, the density and flux of the quanta do not form a four-vector, so that these quantities
 may surely be given no simple physical meaning.'' (see p.734 in \cite{OppiPHOTON}) 
\end{quotation}

\noindent
 Oppenheimer's last point is often restated as the impossibility to construct a conserved
probability from a  \emph{quantum-mechanical} spin-1 wave equation for a photon
\cite{LandauLifshitzBOOKrelQM,BohmBOOKonQM,WeinbergBOOKqft}.
 For instance, here is Bohm (\cite{BohmBOOKonQM}, p.98.): ``There is, strictly speaking, no function which represents the  
probability of finding a light quantum at a given point;'' and here are Berestickii et al.
(\cite{LandauLifshitzBOOKrelQM}, p.3.): 
``[T]he photon wave function cannot be used to construct a quantity which might serve as a probability density
satisfying the necessary conditions of relativistic invariance.'' 


 In this paper we will argue that none of the above cited opinions is correct: it is feasible to formulate
a Lorentz-covariant quantum mechanical wave equation, and it is \emph{not} ``just [the system of] Maxwell's
free-space equations for left- and right-circularly polarized radiation.'' 
 Instead, we propose that the following Lorentz-covariant Dirac(-type) wave equation, 
\beq \label{eq:DiracEQphoton} 
\boxed{\ga^\mu p_\mu \psiPH  + \mPH c\,\Pi \psiPH= 0}
\eeq
with $\psiPH$ a rank-two bi-spinor field, $\Pi$ a projection onto its ``non-mixed'' part,
and with $\mPH>0$ a mass parameter which will \emph{not} be determined in this paper,\footnote{The parameter $\mPH$ is {needed} 
  to ensure that the two terms summed in \refeq{eq:DiracEQphoton} have the same dimensionality; its \emph{physical significance}
  is presumably determined by an interacting theory.
  However, we will prove that solutions $\psiPH$ to \refeq{eq:DiracEQphoton} also satisfy the massless Klein--Gordon wave 
  equation, which means that the photon is indeed massless in this theory.}
satisfies all the requirements expected of the quantum-mechanical wave equation for a single free photon, 
summarized in the following.
 
From the works of Oppenheimer \cite{OppiPHOTON}, Wigner \cite{Wig1939}, Bargmann--Wigner \cite{BarWig1948}, 
Penrose \cite{Pen1976}, and others, we have extracted a list of
viability requirements for any proposed (special-)relativistic single-photon quantum-mechanical wave function and wave equation.
 There is a general part and a photon-specific part.
\smallskip

 The general requirements apply to {\em any} elementary particle:
\begin{enumerate}
\item The wave equation needs to be covariant will respect to the {\em full} Poincar\'e group of the Minkowski spacetime, consisting 
of continuous space and time translations, spatial rotations, and Lorentz boosts; 
furthermore (unless parity violation is a feature of the theory) spatial reflection (parity transformation) and time-reversal.
\item The wave equation needs to be complex linear, and first order in time and space derivatives.
\item The wave equation needs to be derivable from a scalar Lagrangian, in  order to guarantee existence of Noetherian conservation 
laws (energy, momentum,  angular momentum, centroid location) associated with symmetries of the domain, as well as ``charge''-type 
conserved currents associated with any ``internal" symmetries (such as those of the fiber.)
\item The wave equation must yield a conserved {\em quantum probability current} that transforms like a Lorentz 
vector, is future-directed, and (at least generically) timelike.
 Moreover, this quantum probability must be compatible with Born's rule.
\end{enumerate}

 The photon-specific requirements are:
\begin{enumerate}
\item The wave function must transform like a spin-one representation of the Poincar\'e group.
\item The wave equation must imply a massless dispersion relation.
\item The solutions of the wave equation must have no longitudinal modes.
\item The theory should allow for wave functions of both chirality (left-handed and right-handed waves), without having to change 
the sign of the energy.
\end{enumerate}

 We will show in detail that our wave equation \refeq{eq:DiracEQphoton} satisfies all of the above requirements.
 In particular, we claim that we have overcome the two main obstacles which so far have been perceived as insurmountable:
we show that \refeq{eq:DiracEQphoton} does derive from a relativistically-scalar Langrangian, 
and that it furnishes a conserved probability  ``of finding a light quantum at a given point'' on any spacelike hypersurface 
of the configuration spacetime,  which transforms like a time component of a four-vector, and is
compatible with Born's rule \cite{BornsPSISQUAREpapersA, BornsPSISQUAREpapersB}, \cite{Dirac1930}.

 Earlier proposals by other authors also involved a Dirac-type equation for ``a photon wave function;'' 
see Harish-Chandra \cite{HC1946}, Penrose \cite{Pen1976}, and Bia{\l}ynicki-Birula \cite{BiBiTHREE,IBBphotonREV}. 
 Harish-Chandra \cite{HC1946} found ``half'' of what we propose to be the correct equation;\footnote{More precisely, 
  Harish-Chandra works with a ten-dimensional spin-one representation given by Kemmer's matrices \cite{Kem1939}, 
  which is strictly contained in our representation $\psiPH$.
  Harish-Chandra's work seems to be not widely known.
  However, it can be found almost verbatim in \cite{Cor1953};
  also Struyve's Ph.D. thesis \cite{StruyvePHD} includes a brief review of Harish-Chandra's paper on pp. 31-33.}
in particular, Harish-Chandra does realize that a zero-mass spin-one equation is generally not obtained by taking the zero-mass limit of
a massive spin-one equation.\footnote{\label{fn:mzapp}Incidentally, one won't find a mass parameter ($\mPH$ in our notation) in \cite{HC1946} --- 
revealing his mathematically minded personality, Harish-Chandra set all physical parameters $=1$. 
  Struyve's summary of Harish-Chandra's paper features a mass parameter ($m$ in his notation) in front of Harish-Chandra's projection operator.
  Struyve notes that this is not the mass of the spin-1 quantum particle by proving that each component of Harish-Chandra's photon wave function 
 satisfies the massless Klein--Gordon equation.}
Penrose's proposal \cite{Pen1976} involves a {\em pair} of electromagnetic field tensors, one for left-handed and the other 
for right-handed photons, he emphasizes that they should not be complex conjugates of each other, and writes the Maxwell 
equations satisfied by each one in spinorial form (thus coming close to a Dirac equation-type formulation.)  
However, he does not indicate how they should be combined into a system of equations for a {\em single} object, thus
missing Harish-Chandra's projection operator.
 Bia{\l}ynicki-Birula's equation \cite{BiBiTHREE,IBBphotonREV} (with a precursor in the work of Oppenheimer \cite{OppiPHOTON}) also
``doubles up'' the representation in order to handle chirality, yet involves the self-duality constraint which reduces the description 
once again to half the needed number of terms. 
 Moreover, Bia{\l}ynicki-Birula misses Harish-Chandra's projection term as well.
 In none of these proposals have the issues of conserved energy-momentum and probabilities been handled satisfactorily.
 
 The rest of our paper is structured as follows:
 In Section \ref{sec:PSIphoton}, we introduce our photon wave function $\psiPH$.
 In Section \ref{sec:qmwe} we explain our proposal for the photon wave equation \refeq{eq:DiracEQphoton} 
that $\psiPH$ satisfies (see subsection \ref{sec:QMWE}), and we demonstrate its Lorentz covariance, as well as its gauge invariance ---
see subsections \ref{sec:lorentzcov} and \ref{sec:gaugeinv}, respectively. 
  Our claims about the dispersion relation and transversality are proven in Section \ref{sec:PDE}.  
 Section \ref{sec:conservation} derives all Noether-type conservation laws for the solutions of the photon wave equation from 
its Lagrangian.
  In Section \ref{sec:riesz} we obtain a different set of conservation laws for this equation, not based on its Lagrangian but
on the M. Riesz tensor, and we show how these conserved quantities can be used to construct a quantum probability current for the 
photon.
 In Section \ref{sec:HilbertBorn} the conserved probability will guide us to extract from the relativistic 
photon wave equation a Schr\"odinger equation with photon Hamiltonian, and a conventional $L^2$ inner product Hilbert space formalism,
which shows that our quantum probability is  compatible with Born's rule.  
 Incidentally, also the Einstein relations are extracted from the Hamiltonian.
 The main body of our paper concludes in Section \ref{sec:conclusions} with a summary, and with an outlook on future work.
 There are two appendices. 
 In Appendix \ref{app:others} we contrasts our bottom-up approach with some top-down approaches, popular with the larger 
quantum optics community, which yield ``photon wave functions `for many practical purposes';'' 
see \cite{IBBphotonREV, ScullyBOOK, SmithRaymer, Chandrasekar, Hawton, Mostafazadeh}.
 In Appendix \ref{app:primer} we provide a brief recap of the algebra of spinors and bi-spinors of ranks one and two, 
 for the convenience of the reader. 

\section{The quantum-mechanical wave function of a single photon}\label{sec:PSIphoton}
 In this section we present our choice for the quantum-mechanical wave function of a single photon. 
 For the convenience of the reader, in Appendix \ref{app:primer} we have collected some notational 
choices and pertinent facts on spinors and bi-spinors, which we invite the reader to consult as needed.


We propose the photon wave function to be a {\em rank-two bi-spinor} (see subsection \ref{sec:defPWF}) with trace-free diagonal blocks,
namely
\beq \label{defpsiPH}
\psiPH := \left(\begin{array}{cc} \psi_+ & \chi_-\\ \chi_+ & \psi_- \end{array}\right)\quad \mbox{with}\quad
\tr(\psi_+) = \tr(\psi_-) = 0.
\eeq
 It follows that there are two complex-valued one-forms $\bfa_+$ and $\bfa_-$, and two real-valued two-forms $\bff_+$ and $ \bff_-$ on the 
configuration space-time $\Rset^{1,3}$ such 
that 
\beq \label{Fpm} \chi_+ = \si'(\bfa_+),\quad \chi_-= \si(\bfa_-),\qquad \psi_+ = \Si(\bff_+),\quad
\psi_- = \Si'(\bff_-).
\eeq
(See (\ref{def:si}, \ref{def:sipx}, \ref{def:Siff}) for the definitions of the mappings $\si$, $\si'$, $\Si$, and $\Si'$.) Thus, $\psiPH$ so defined is a particular form of a rank-two bi-spinor. 
 The set of all such $\psiPH$ clearly forms a complex-linear sub-bundle $\cT$ of the bundle $\cB$ of all rank-two bi-spinors.  

Consider the subset $\cS$ of $\cT$ consisting of all $\psiPH$ of the form \refeq{defpsiPH} for which $\bff_+ = \bff_-$ and 
$\bfa_+ = \bfa_-$.
  It is easy to see that $\cS$ is a {\em real}-linear subspace of rank-two bi-spinor space, not a complex-linear one.
  Namely, according to the discussion in Appendix \ref{app:primer}, there must exists a $\xiV \in \Cset^4$ with $\xi^0=0$ such that
$\psi_+ = i \xi^\mu \si_\mu$ and $\psi_- = -i  \xi^\ast{}^\mu \si_\mu$, and this relation is not preserved under multiplication by a 
complex number. 
 We refer to $\cS$ as the subspace of {\em self-dual} bi-spinors.  

Finally, we define the projection operator $\Pi$, whose action on the photon wave function $\psiPH$ amounts to killing off the off-diagonal 
terms $\chi_\pm$:
\beq \label{def:Pi}
\Pi \psi 
:= \Pi_{+}^{} \psi \Pi_{+}^{} + \Pi_{-}^{} \psi \Pi_{-}^{},
\eeq
where $\Pi_\pm$ are the projection operators defined in \refeq{def:projections}.

\section{The quantum-mechanical wave equation for a single free photon}\label{sec:qmwe}
%
 We now explain our proposal that the  \emph{quantum-mechanical} wave equation for the wave
function of a single free photon in physical Minkowski spacetime $\Rset^{1,3}$ is given by  
the Dirac(-type) equation \refeq{eq:DiracEQphoton} for a rank-two bi-spinor field $\psiPH$ on configuration space-time $\Rset^{1,3}$ 
(a copy of Minkowski spacetime), and we demonstrate its covariance under the Lorentz group, and its gauge invariance.

\subsection{Formulation of the wave equation for a single free photon}\label{sec:QMWE}

 To pave the ground for our ``Dirac-type wave equation for the photon,'' we begin by recalling Dirac's wave equation \refeq{eq:DiracEQelectron} 
for a free electron of empirical rest mass $\mEL$ \cite{Dirac1928a,Dirac1928b},\!~viz.\footnote{Here,
as usual, $c$ denotes the speed of light in vacuum, and $\hbar$ denotes the ``reduced'' Planck quantum of action.
 We remark that we could streamline the notation by choosing spacetime units such that $c=1$ and, given that, also
energy units such that $\hbar=1$, 
but we found it helpful to retain both $c$ and $\hbar$ when setting up the model.} 
\beq \label{eq:DiracEQelectronAGAIN} 
{ -i\hbar \ga^\mu \partial_\mu \psiEL + \mEL c \Id \psiEL =0}.
\eeq
 The wave function $\psiEL$ in \refeq{eq:DiracEQelectronAGAIN} is a complex $4\times 1$ matrix-valued unknown function on $\Rset^{1,3}$, 
a rank-one bi-spinor field, which encodes the physicists' classification of the electron as being a ``massive spin-$\frac12$ particle.''  
 We recall that solutions $\psiEL$ of \refeq{eq:DiracEQelectronAGAIN} also satisfy the massive Klein-Gordon equation
\beq \label{eq:massivKGeq}
\dal\psiEL + \tfrac{\mEL^2c^2}{\hbar^2}\psiEL =0.
\eeq

 A photon is classified as a ``massless spin-1 particle,'' and so a Dirac-type wave equation for a photon could be expected 
to differ from \refeq{eq:DiracEQelectronAGAIN} in \emph{at least} two ways: 

\ben

\item 
{The wave function ($\psiPH$, here) of a spin-$1$ particle needs to transform like a {\em bi-spinor field of rank two}
under the action of a Lorentz tranformation $\bLa \in O(1,3)$ (cf. \cite{PenroseRindlerBOOK}), viz.
\beq 
\psiPH \mapsto \bL \psiPH \bL^{-1},
\eeq
where $\bL$ is the projective spinorial representation of $\bLa$ (see our primer on spinors, Appendix \ref{app:primer}) 
--- replacing $\psiEL$ by
such a $\psiPH$ in  \refeq{eq:DiracEQelectron} indeed yields a relativistic wave equation for a spin-1 particle of rest mass $\mEL$;}

\item
{Since the photon is (rest-)massless, one would expect that this means to simply set $\mEL\to 0$ 
in \refeq{eq:DiracEQelectronAGAIN} (after replacing $\psiEL$ by a {bi-spinor field of rank two} $\psiPH$),
so that the ``mass term'' is absent from the photon wave equation.
 However, it follows from the work of Harish-Chandra \cite{HC1946} that a massless spin-1 equation cannot be obtained 
from the equation of a  spin-1 particle with rest mass $m$ by simply setting $m$ in the latter equation equal to zero. 
 Rather, in the lower-order term ($\propto \mEL$ in \refeq{eq:DiracEQelectronAGAIN}) one
needs to replace $\Id$ with a non-trivial \emph{projection} operator $\Pi$ (and $\mEL$ by a yet to-be-determined mass parameter $\mPH>0$ ---
we refrain from writing $\mPHO$ to avoid the inadvertent suggestion that the photon had a rest mass.)
 For a photon this projector is the orthogonal projection onto the subspace of block-diagonal rank-two bi-spinors (see subsection 
\ref{sec:defPWF} for  details.)}

\een

\noindent
 Hence our proposal, that the quantum-mechanical wave equation for a free photon is of the following Dirac-type,
precisely:
\beq \label{eq:DiracEQphotonAGAIN} 
{-i\hbar \ga^\mu \partial_\mu \psiPH  + \mPH c\,\Pi \psiPH= 0},
\eeq
with $\psiPH$ a trace-free rank-two bi-spinor field of a massless spin-1 particle, and with $\mPH>0$ a dimensional constant (whose value will not be determined in this paper.)
 We emphasize that despite the appearance of $\mPH$ in (\ref{eq:DiracEQphotonAGAIN}) we will find
that solutions $\psiPH$ of \refeq{eq:DiracEQphotonAGAIN} also satisfy the massless Klein-Gordon equation 
\beq \label{eq:masslessKGeq}
\dal\psiPH=0.
\eeq
 Moreover, we will see that $\Pi\psiPH$ satisfies a massless Dirac equation. 

 We note that the set of solutions to \refeq{eq:DiracEQphotonAGAIN} clearly forms a complex-linear vector space. 

Below, we will first show the covariance of \refeq{eq:DiracEQphotonAGAIN} under the action of the {\em full} Poincar\'e group; 
then we will formulate the gauge transformations under which it is invariant.

\subsection{Poincar\'e covariance of the photon wave equation}\label{sec:lorentzcov}
%

%
Since \refeq{eq:DiracEQphoton} is a constant coefficient differential equation, it is clearly covariant with respect to time- and space-translations.  We thus only need to check for Lorentz covariance.   We can easily show that  (\ref{eq:DiracEQphoton}) is covariant 
under the {\em full} Lorentz group, just as (\ref{eq:DiracEQelectron}) is:
 Let $\bLa \in O(1,3)$ be an element of the full Lorentz group and let $\bL \in \fL$ be its projective spinorial representation (see subsection \ref{sec:LorentzRep}.)
 Suppose $\psiPH$ is a solution of \refeq{eq:DiracEQphoton}.
 We show that
\beq 
\psiPH'(x) := \bL \psiPH( \bLa^{-1} x) \bL^{-1}
\eeq
is a solution of the same equation.  Let $x' := \bLa^{-1} x$.  
 From the definition of $\Pi$ it easily follows that
\beq 
\label{PiLcomm} 
\Pi \bL \psiPH(x')  \bL^{-1} = \bL \Pi \psiPH(x') \bL^{-1}.
\eeq
Since by assumption $\psiPH$ is a solution of (\ref{eq:DiracEQphoton}), we have
\beq 
 -i \hbar \ga^\mu \frac{\p}{\p {x'}^\mu} \psiPH(x') + \mPH c\Pi \psiPH(x') = 0.
\eeq
Multiplying on the left by $\bL$ and on the right by $\bL^{-1}$, and using \refeq{PiLcomm} we obtain
\beq 
-i\hbar \bL \ga^\mu \frac{\p x^\la}{\p {x'}^\mu} \frac{\p}{\p x^\la} \psiPH(x') \bL^{-1} + \mPH c \Pi \bL\psiPH(x')\bL^{-1} = 0,
\eeq
which implies
\beq 
-i \hbar \bL \bLa^\la_\mu \ga^\mu \bL^{-1}\bL \p_\la \psiPH(x') \bL^{-1} + \mPH c \Pi \bL\psiPH(x')\bL^{-1} = 0.
\eeq
Recalling that the Lorentz-vector-valued rank-2 bi-spinor $\ga^\mu$ transforms as $\ga^\la = \bL \bLa^\la_\mu \ga^\mu \bL^{-1}$, we thus have
\beq 
-i \hbar \ga^\la \p_\la \psiPH'(x) + \mPH c \Pi \psiPH'(x) = 0,
\eeq
establishing the covariance under the full Lorentz group.
%
\subsection{Gauge invariance of the photon wave equation}\label{sec:gaugeinv}
%
 The Dirac-type equation \refeq{eq:DiracEQphoton} is invariant under gauge transformations 
\beq \label{gaugetrans}
\psiPH \mapsto \psiPH' = \psiPH + (\Id - \Pi) \Upsilon,
\eeq
where $\Upsilon$ is a rank-two bi-spinor satisfying a so-called ``massless Dirac equation,'' viz.
\beq \label{eq:DIRACmZEROeqnGAUGE}
\ga^\mu p_\mu \Upsilon = 0.
\eeq
(Obviously, only the off-diagonal blocks of $\Upsilon$ are significant.)  
This can be readily seen by substituting $\psiPH'$ in \refeq{eq:DiracEQphoton} and using the following easily-verifiable identity
\beq \label{iden:gaPi}
\ga^\mu (\Id - \Pi) = \Pi \ga^\mu,\qquad \mu = 0,\dots,3.
\eeq

 It is also easy to see that $\Pi\psiPH$, the {\em gauge-invariant part} of $\psiPH$, satisfies \refeq{eq:DIRACmZEROeqnGAUGE} as well, i.e.
\beq \label{eq:DIRACmZEROeqn}
\ga(p)\Pi\psiPH = 0, 
\eeq
for,
\beq 
\ga(p)\Pi\psiPH = (\Id-\Pi)\ga(p)\psiPH =- \mPH c(\Id - \Pi)\Pi\psiPH = 0.
\eeq

 The gauge-invariance of the diagonal part of the photon wave function, as established in the above,  together with the common 
classification of the photon as a ``mass-zero particle,'' 
would seem to suggest that we could work directly with $\Pi\psiPH$ and \refeq{eq:DIRACmZEROeqn} instead of
 $\psiPH$ and \refeq{eq:DiracEQphoton}.
 However, the off-diagonal part will turn out to be vital for the derivation of \refeq{eq:DiracEQphoton}, and thus of \refeq{eq:DIRACmZEROeqn},
from a Lorentz-invariant Lagrangian.
 The off-diagonal part may also play a role in the incorporation of the interaction of a photon with an electron. 
In Section~\ref{sec:conservation} we will find some novel conservation laws involving both the diagonal and the off-diagonal components 
of $\psiPH$, indicating that the off-diagonal parts cannot be ignored. In subsection \ref{sec:offdiag} we will take a closer look at 
these off-diagonal terms and the equations they satisfy.  
\section{The photon wave equation as a partial differential equation}\label{sec:PDE}
%
 Several features expected of a genuine quantum-mechanical wave equation of a photon follow from the fact that \refeq{eq:DiracEQphoton}
is a partial differential equation. 
 Here we will show that:
\begin{enumerate}
\item 
the dispersion relation of this photon wave equation reads $\omega^2 =  c^2|\kV|^2$;
\item
the diagonal part of solutions to the photon wave equation are transversal.
\end{enumerate}

\subsection{The zero-mass dispersion relation}
%
By applying $\ga(p)$ to  \refeq{eq:DiracEQphoton} and using \refeq{iden:gaPi}, we find
\beq 
\ga(p)\ga(p)\psiPH = - \ga(p)\Pi \psiPH = - (\Id - \Pi) \ga(p) \psiPH = \mPH c (\Id - \Pi)\Pi \psiPH = 0.
\eeq
 Now, $\ga(p)^2 = \Id p_\mu p^\mu  = - \hbar^2 \Id \dal$, and so it follows that each component of a solution $\psiPH$ of \refeq{eq:DiracEQphoton} 
also satisfies the classical wave equation, viz. the partial differential equation
\beq \label{eq:waveEQ}
\dal\psiPH=0;
\eeq
in a quantum-mechanical context \refeq{eq:waveEQ} is known as the massless Klein--Gordon equation.
 This establishes (\ref{eq:DiracEQphoton}), respectively (\ref{eq:DiracEQphotonAGAIN}),
 as a ``zero-mass'' wave equation for all its components, 
despite the occurrence of $\mPH$ in (\ref{eq:DiracEQphoton}), (\ref{eq:DiracEQphotonAGAIN}).
 
 In particular, the dispersion relation for a photon wave function satisfying (\ref{eq:DiracEQphoton}) reads
\beq \label{dispREL}
\omega^2 = c^2|\kV|^2,
\eeq
which agrees with the one expected for a zero rest-mass particle, such as the photon. 
%
\subsection{Transversality (of the diagonal components) of the photon wave function} 
%
We now establish that $\Pi\psiPH$ propagates only in transversal modes.
 Namely, like (\ref{eq:DiracEQelectron}), so also (\ref{eq:DiracEQphoton}) can be written 
as a system of two equations for two spinors, rank-two spinors now.
  Choosing a Lorentz frame as before, and  the Weyl representation \refeq{def:gammas}, 
our Dirac equation for the photon (\ref{eq:DiracEQphoton}) decomposes into
\beq 
\begin{array}{lcl}\label{DIRACeqSPLITdiag}
\phantom{\hbar} (\p_t + c\siV\cdot\nab)\si(\xiV_+) \!\!&=&\!\! 0, \\
\phantom{\hbar} (\p_t - c\siV\cdot\nab)\si'(\xiV_-^*) \!\!&=&\!\! 0,
\end{array}
\eeq
for the diagonal components, and into
\beq 
\begin{array}{lcl}\label{DIRACeqSPLIToffDIAG}
\hbar (\p_t - c\siV\cdot\nab) \si'(\bfa_+) & = & \mPH c^2\si (\xiV_+),\\
\hbar (\p_t + c\siV\cdot\nab) \si (\bfa_-) & = & \mPH c^2\si'(\xiV_-^*),
\end{array}
\eeq
for the off-diagonal ones; here, $\siV$ is the three-vector of the Pauli matrices.

 Now it is well-known \cite{LapUhl1931,OppiPHOTON,IBBphotonREV} that if one sets 
\beq \label{xiISePLUSib}
\xiV = \begin{pmatrix} 0 \\  \be + i \bb \end{pmatrix},
\eeq
for (column-)vector fields $\be:\Rset^3\to\Rset^3$ and $\bb:\Rset^3\to\Rset^3$, then 
\vspace{-5pt}
\beq \label{eq:Weyl}
(\p_t + c\siV\cdot\nab) \si(\xiV) = 0
\eeq
is \emph{formally} equivalent to the Maxwell system of equations for a source-free (formal) electric field $\be$ and a (formal) 
magnetic induction field $\bb$ in the given Lorentz frame, viz.
\beq \label{MaxEB}
\begin{array}{ll}
&\p_t\be - c\nab\times\bb = 0, \qquad \nab\cdot \be = 0,\\
&\p_t\bb + c\nab\times\be = 0, \qquad \nab\cdot \bb = 0.
\end{array}
\eeq
 Thus, the first equation in (\ref{DIRACeqSPLITdiag}) with $\xiV_+ = \xiV$ as in (\ref{xiISePLUSib}), and the second equation
in (\ref{DIRACeqSPLITdiag}) with $\xiV_- =\xiV$ as in (\ref{xiISePLUSib}), each are equivalent to (\ref{MaxEB}).
 This proves absence of longitudinal modes in (\ref{DIRACeqSPLITdiag}).

\begin{rem}\textit{
There is an important distinction between \refeq{eq:Weyl} and \refeq{MaxEB} as field equations on $\Rset^{1,3}$:
  It is well-known that \refeq{MaxEB} is covariant under the {\em full} Lorentz group if one assumes that $\be$ and $\bb$
transform like usual electric and magnetic Maxwell fields.
 Likewise, equation \refeq{eq:Weyl} is covariant under the proper Lorentz group if one assumes that the space part of $\xiV$ 
transforms like $\be+i\bb$, but it would still not be covariant under the full Lorentz group; specifically, it is not covariant 
under parity transformations (space reflections / inversions), even if we assume that under such a reflection $\xiV$ goes to 
$-\xiV^\ast$ (as is suggested by $\be+i\bb\mapsto-\be +i\bb$ under space inversions; recall that the electric $\be$ is a polar 
and the magnetic $\bb$ an axial vector field) --- a space inversion then transforms (\ref{eq:Weyl}) into 
\beq \label{eq:WeylREFLECT}
(\p_t - c \siV\cdot\nab) \si(P\xiV^\ast) = 0,
\eeq 
with $P$ as in \refeq{def:Pinver}. It follows that, to implement the full Lorentz group using a spinorial rank-two representation,
at least the \emph{pair} of equations \refeq{eq:Weyl},\refeq{eq:WeylREFLECT} would be needed, corresponding to the fact that 
\refeq{MaxEB} allows superpositions of both left- and right-handed waves; cf. {\rm \cite{Pen1976}, \cite{IBBphotonREV}}.
 Thus, in a given Lorentz frame, the pair \refeq{eq:Weyl},\refeq{eq:WeylREFLECT} would then be equivalent
to (\ref{DIRACeqSPLITdiag}) \emph{plus the self-duality constraint} $\xiV_- =\xiV_+$.}

\textit{ \emph{However}, $\xiV$ as a \emph{four vector} which in a particular Lorentz frame is identified with $\be + i\bb$ 
through (\ref{xiISePLUSib}) in that frame does not transform into a four vector $\xiV'$ whose space part is given by $\be' +i\bb'$ 
and time component by $0$, where $\be'$ and $\bb'$ are the images of $\be$ and $\bb$ in the new frame.
 This demonstrates that the photon wave equation on the single-photon configuration spacetime $\Rset^{1,3}$ is \emph{not} 
equivalent to the Maxwell field equations on physical spacetime $\Rset^{1,3}$.}
\end{rem}

\subsection{The off-diagonal components of the photon wave function.}\label{sec:offdiag}
 Recall from \refeq{Fpm} that 
\beq
\chi_+ = \si'(\bfa_+),\qquad \chi_- = \si(\bfa_-),
\eeq
where $\bfa_\pm$ are complex-valued 1-forms on the configuration Minkowski space. 
 Let us choose a Lorentz frame as before and denote the components of $\bfa_\pm$ in that frame in the following way:
\beq\label{eq:insidechi}
\chi_+ = \varphi_+\si_0 - \siV \cdot \ba_+,\qquad \chi_- = \varphi_- \si_0 + \siV \cdot \ba_-,
\eeq
 for $\varphi_\pm \in \Cset$ and $\ba_\pm \in \Cset^3$. 
 Let us also recall that, in the same frame, we have
\beq\label{eq:insidepsi}
\psi_+ = i \siV \cdot (\be_+ + i \bb_+),\qquad \psi_- = - i \siV \cdot (\be_- - i \bb_-).
\eeq
  Let $\psiPH$ be a solution of \refeq{eq:DiracEQphoton}. 
 Writing \refeq{DIRACeqSPLIToffDIAG} out in components we obtain
\beq\label{eq:potentials}
\tfrac{1}{c}\p_t \varphi_\pm + \nab \cdot \ba_\pm = 0,
\qquad  
\be_\pm = \tfrac{\hbar}{\mPH c} \left(-\nab \varphi_\pm - \tfrac{1}{c}\p_t \ba_\pm\right),
\qquad 
\bb_\pm = \tfrac{\hbar}{\mPH c} \nab\times \ba_\pm.
\eeq
It thus appears that the relationship of the off-diagonal terms in the photon wave function 
 $\psiPH$ to its diagonal terms is formally the same as that of {\em electromagnetic potentials} 
(in Lorentz gauge) to their corresponding electromagnetic fields.  

Equations \refeq{DIRACeqSPLITdiag} also imply that $\varphi_\pm$ and $\ba_\pm$ must satisfy the classical wave equation: 
\beq\label{eq:waves}
\tfrac{1}{c^2} \p_t^2 \varphi_\pm - \Delta \varphi_\pm = 0,\qquad \tfrac{1}{c^2} \p_t^2 \ba_\pm - \Delta \ba_\pm = 0.
\eeq

Even though the off-diagonal terms appear to have 8 complex, or 16 real degrees of freedom, 
it turns out that half of those are due to gauge freedom. 
 More precisely, we have the following
\begin{prop}
Let $\psiPH = \left(\begin{array}{cc} \psi_+ & \chi_-\\ \chi_+ & \psi_- \end{array} \right)$ be a solution of \refeq{eq:DiracEQphotonAGAIN}. 
 There exists a gauge transformation $\psiPH \mapsto \psiPH + (\Id - \Pi)\Upsilon$, with $\Upsilon$ satisfying \refeq{eq:DIRACmZEROeqnGAUGE},
 such that after applying it, the $\chi_\pm$ are Hermitian matrices (equivalently, 
the $\bfa_\pm$ are real-valued.)
\end{prop}
\begin{proof}
 We will prove the claim for $\chi_+$. 
 The argument for $\chi_-$ is identical. 
 To simplify notation we will henceforth drop the ``$+$" subscript. 
We know that $\chi$ satisfies the following equation:
\beq\label{eq:chi}
-i\hbar \si(\p) \chi + \mPH c \psi = 0.
\eeq
Let $\varphi= \varphi_R + i\varphi_I$ and $\ba =\ba_R + i \ba_I$  be the decomposition into real and imaginary parts of $\varphi$ and $\ba$. 
 Suppose we can show there exist real-valued functions $\tilde{\varphi}$ and $\tilde{\ba}$ such that 
\bna
-\nab \tilde{\varphi} - \p_{ct} \tilde{\ba} & = & - \nab \varphi_R - \p_{ct} \ba_R - \nab \times \ba_I \nonumber\\
\nab \times \tilde{\ba} & = & \nab \times \ba_R - \nab \varphi_I - \p_{ct} \ba_I \label{eq:syspot} \\
\p_{ct} \tilde{\varphi} + \nab \tilde{\ba} & = & 0. \nonumber
\ena
 If the above holds, it is then obvious from \refeq{eq:potentials} that 
\beq
\be+i \bb = \tfrac{\hbar}{\mPH c}\left[ (-\nab \tilde{\varphi} - \p_{ct} \tilde{\ba}) + i ( \nab \times \tilde{\ba})\right].
\eeq
 Therefore, if we let $\tilde{\chi} := \tilde{\varphi} \si_0 - \siV \cdot \tilde{\ba}$ then $\tilde{\chi}$ is Hermitian and we have 
\beq
-i\hbar \si(\p)\tilde{\chi} + \mPH c \psi = 0.
\eeq
Subtracting \refeq{eq:chi} from the above equation we obtain that $\upsilon := \tilde{\chi} - \chi$ is a gauge, which establishes the claim.

To prove that \refeq{eq:syspot} has a solution it is enough to check that, by virtue of \refeq{eq:waves}, both $\tilde{\varphi}$ and 
$\tilde{\ba}$ satisfy the classical wave equation, for which real-valued solutions with initial data compatible with \refeq{eq:syspot} 
can be easily found. 
\end{proof}

\section{Lagrangian formulation and Noetherian conservation laws}\label{sec:conservation}
 In this section we show first that our Dirac-type wave equation for the photon \refeq{eq:DiracEQphotonAGAIN} 
derives from a Lorentz-scalar action principle.
 We next derive the full set of Noetherian conservation laws for (\ref{eq:DiracEQphoton}) that follow from this Lagrangian.
 Beside the anticipated laws of energy, momentum, and angular-momentum conservation, also included are some novel laws which express
some form of cross-chirality conservation for the left-handed and right-handed parts of the photon wave function. 
 
\subsection{The action principle}\label{sec:action}
The Dirac equation \refeq{eq:DiracEQphoton} is the Euler--Lagrange equation for an action functional with (real) scalar Lagrangian density given by
\beq \label{Lagrangian}
\cLph'=  
 \frac{\hbar c}{16\pi i} \tr\left({\psiPHb} \gamma^\mu_{} \p_\mu^{} \psiPH - \p_\mu \psiPHb \ga^\mu \psiPH\right) 
+ \frac{\mPH c^2}{8 \pi} \tr \left( \psiPHb \Pi \psiPH \right),
\eeq
where the Dirac adjoint $\psiPHb$ and the trace operation $\tr$ were defined in subsection \ref{sec:clifford} 
(See \refeq{def:Diradj} and Remark \ref{rem:diradj}.)

 We  observe that $\cLph' \equiv 0$ when evaluated on any solution of the Euler--Lagrange equations \refeq{eq:DiracEQphoton}. 

 One notes that the above Lagrangian is not gauge-invariant.  
 This can be remedied in the following way.  
 The Lagrangian density, defined by
\beq \label{gaugeinvlag}
8\pi \cLph:= -i \hbar c \tr \left( \overline{\Pi\psiPH} \ga^\mu \p_\mu \psiPH - \p_\mu\psiPHb \ga^\mu \Pi \psiPH \right) 
+ \mPH c^2 \tr \left( \psiPHb \Pi \psiPH \right)
\eeq
is invariant under gauge transformations \refeq{gaugetrans} and gives rise the to same Euler--Lagrange equation \refeq{eq:DiracEQphoton},
as can be easily checked using \refeq{iden:gaPi}. 
 We note that, unlike $\cLph'$ however, this Lagrangian density does not vanish along solutions of the Euler-Lagrange equations,
thus leading to more complicated expressions in the forthcoming calculations of conservation laws.

\subsection{Noether symmetries and conservation laws}
 If one views $\psiPH$ as a section of the spin bundle $\cT$ whose base manifold is the configuration space(time) $\Rset^{1,3}$ (a copy 
of Minkowski spacetime) and whose fibers are rank-two bi-spinors of the form \refeq{defpsiPH}, then according to Noether\rq{}s 
Theorem \cite{Noe1918}  (see \cite{Chr2000} for precise statements), to every flow on $\cT$ that preserves the Lagrangian \refeq{Lagrangian} 
there corresponds a quantity that is conserved along any solution of the corresponding Euler--Lagrange equations. 
 A flow on $\cT$ meanwhile, could be a ``vertical" flow that is purely in the fiber, i.e. it does not move the base point, or a ``horizontal" one 
that moves the base point, and may also  have a nontrivial action on the fibers in order to maintain covariance. 
  We will look at both types below.

\subsubsection{Conservation laws due to symmetries of the fibers (helicity and such)}
 Let $G$ be a section of the bundle $\cB$ of rank-2 bi-spinors over the configuration Minkowski space, and consider the flow generated by the right 
action 
of $G$ on rank-two bi-spinors $\psi \in \cB$:
\beq 
\psi_s := \psi e^{isG},\qquad s \in (-\ep,\ep).
\eeq
 The vector field whose integral curves give the flow lines of this flow is
\beq 
Z := \left.\frac{d}{ds}\right|_{s=0} \psi_s = i \psi G.
\eeq
We look for conditions on $G$ such that the above flow preserves the Lagrangian \refeq{Lagrangian}. 
 In order to do so, the flow must preserve the sub-bundle $\cT$ of rank-2 bi-spinors with trace-free diagonal blocks.  
It follows that one must have 
\beq  \tr (\psiPH G ) = 0,
\eeq for all $\psiPH \in \cT$.  This condition severely restricts $G$, and it is not hard to check that $G$ can only be of the form
\beq \label{Gdiag} G = \left(\begin{array}{cc} \la_1 \Id_2 & \\ & \la_2 \Id_2 \end{array}\right),\qquad \la_1,\la_2 \in \Cset.
\eeq
On the other hand, it is easy to see that for any $G$, $\overline{e^{isG}} = e^{-is\overline{G}}$. 
 So
\beq 
 \tr\left( e^{-is\overline{G}} [ \psiPHb\ga^\mu \p_\mu \psi -  \p_\mu \psiPHb \ga^\mu \psi ] e^{isG} \right) =  \tr\left( \psiPHb\ga^\mu \p_\mu \psi -  \p_\mu \psiPHb \ga^\mu \psi \right),
\eeq
provided 
\beq \label{Gselfadj}\overline{G} = G.
\eeq
 Assume this holds. 
 A direct computation shows that for any invertible $2\times 2$ matrix $B$, one has
\beq 
\tr \left( B^{-1} \psiPHb \Pi (\psiPH B)\right) = \tr (B^{-1} \psiPHb (\Pi\psiPH)B) = \tr \left(\psiPHb \Pi \psiPH\right). 
\eeq
 Setting $B = e^{isG}$ one then sees that the Lagrangian \refeq{Lagrangian} will be preserved under the flow of $Z$, and 
therefore the Lie derivative of $\cL$ with respect to $Z$ must be zero. 
 Meanwhile the two restrictions \refeq{Gdiag} and \refeq{Gselfadj} together imply that there exist $a,b\in\Rset$ such that
\beq 
G = a \Id_4 + i b \ga^5.
\eeq

 According to \cite{Chr2000}, the conserved Noether current corresponding to $Z$ is $j_Z^\mu = p^\mu_a Z^a$ where $p^\mu_a$ are 
the {\em canonical momenta} (i.e. the derivative of the Lagrangian density with respect to the canonical velocities.)
 We thus have $\nab_\mu j_Z^\mu = 0 $ where $\nab$ represents the covariant derivative on the spin bundle $\cB$, with respect to 
which both the metric tensor $\boldsymbol{\eta}$ and the Dirac gamma matrices are constant. 
 For the photon Lagrangian \refeq{Lagrangian} we have
\beq \label{helicityCURRENT}
j_Z^\mu = p^\mu_a Z^a = \frac{\p \cLph'}{\p(\p_\mu \psiPH)}(i\psiPH G) =  \tr\Big[\frac{1}{16 \pi i} \psiPHb \ga^\mu (i \psi G)\Big] =
 \frac{1}{16 \pi}\tr\left(\psiPHb\ga^\mu \psi G\right).
\eeq

We thus have a 2 parameter family of conserved currents associated with vertical symmetries of the bundle $\cT$, since we have, for
 each $G$ as in the above,
\beq 
\nab_\mu j_Z^\mu = \nab_\mu \left\{ \frac{1}{16 \pi}\tr\left( \psiPHb \ga^\mu \psiPH G \right) \right\} = 0.
\eeq
 In particular $j_Z^0$, the time component of the current, has the following general form in terms of the photon wave function $\psiPH$:
 There exist  $a,b\in \Rset$ such that 
\bna\label{nongaugeconslaw}
j_Z^0 & = &  \frac{a}{16 \pi} \tr \{ \chi_-^\dag \psi_+ +\psi_+^\dag \chi_- + \chi_+^\dag \psi_- + \psi_-^\dag \chi_+ \} \\
& &+  i\frac{b}{16 \pi} \tr \{ \chi_-^\dag \psi_+ - \psi_+^\dag \chi_- + \chi_+^\dag \psi_- -\psi_-^\dag\chi_+\}.
\ena
 The above expression is of the form of a ``cross-helicity'' of the $\pm$ components of the wave function.
 Helicity currents such as \refeq{helicityCURRENT} are {\em not} gauge-invariant no matter what one chooses for $a$ and $b$;
however, at least in the case $b=0$, integrating $j_Z^0$ over $t=  \mbox{const.}$ hypersurfaces $\Si_t$ (defined by the choice 
of frame $\{\bu_{(\mu)}\}_{\mu = 0}^3$,)  yields a conserved gauge-independent quantity that could be of possible topological significance. 
More precisely, for $a=1$ and $b=0$ we have, using the decompositions \refeq{eq:insidechi} and \refeq{eq:insidepsi} with respect to the chosen 
frame, that
\beq
j_Z^0  =  \frac{1}{8 \pi} (\bb_- \cdot \ba_+ - \bb_+ \cdot \ba_-) = \frac{\hbar}{8 \pi \mPH c}\nab\cdot (\ba_- \times \ba_+).
\eeq
Hence, the integral of $j_Z^0$ on a spatial slice $\Si_t$ is not only conserved in time, it in fact only depends on the leading-order 
asymptotic behavior at spatial infinity of the two formal ``magnetic potentials" $\ba_\pm$, since the integrand is a complete spatial 
divergence.\footnote{In the terminology of Christodoulou \cite{Chr2000} this is a {\em boundary current} for the Lagrangian theory. 
 Well-known examples of boundary currents are the ADM quantities in General Relativity.} 
 Moreover, it is easy to see that the integral will be gauge-independent. 
We can thus define the time- and gauge-independent quantity
\beq
\Xi := \frac{\hbar}{8 \pi \mPH c}\lim_{r\to \infty} \int_{S_r} (\ba_- \times\ba_+)\cdot \bn dS
\eeq
where $S_r$ denotes the sphere of radius $r$ centered at the origin in the spatial slice $\Si_t$. 
 Note that $\Xi$ is a dimensionless quantity.
 When finite, it carries some global (presumably topological)
information about the wave function, and it obviously depends only on its off-diagonal blocks $\chi_\pm$. 

 We will continue to pursue our inquiry into the significance of $\Xi$ in a separate publication. 
 Here instead we focus on obtaining point-wise gauge-invariant conserved quantities.

\subsubsection{Conservation laws due to symmetries of the domain (energy-momentum, etc.)}\label{sec:energymom}
 Each of the two Lagrangians \eqref{Lagrangian} and \eqref{gaugeinvlag} give rise to a {\em canonical stress} 
\cite{Chr2000} through the Legendre transform:
\beq  T_{\mu\nu} = \frac{\p \cLph}{\p (\p_\nu \psiPH)}\p_\mu\psiPH - \eta_{\mu\nu} \cLph.\eeq
One has
\beq \label{Tmunu}
T_{\mu\nu} = \frac{\hbar c}{8 \pi i} \tr \left( \overline{\Pi\psiPH}\ga_\nu\p_\mu \psiPH - \p_\mu \psiPHb \ga_\nu \Pi\psiPH \right) 
+\frac{\mPH c^2}{8\pi} \tr\left( \psiPHb \Pi \psiPH \right) \eta_{\mu\nu}
\eeq
as the canonical stress corresponding to $\cLph$ and
\beq \label{Tprimemunu}
T'_{\mu\nu} = \frac{\hbar c}{16 \pi i} \tr \left( \psiPHb \ga_\nu \p_\mu \psiPH - \p_\mu \psiPHb \ga_\nu \psiPH\right)
\eeq
as the canonical stress corresponding to $\cLph'$.   

These two objects are related: their difference is a complete divergence:
\beq \label{TTprimediff}
T_{\mu\nu} - T'_{\mu\nu} = 
\frac{\hbar c}{16 \pi i}\p^\la\tr\left\{  \psiPHb(\eta_{\mu\nu}\ga_\la - \eta_{\mu\la}\ga_\nu )\Pi\psiPH - \mbox{adj.}\right\}
\eeq
Here and henceforth $\mbox{adj.}$ is shorthand notation for the adjoint of the immediately preceding expression.

It is easy to see that $\nab^\nu T_{\mu\nu} = \nab^\nu T'_{\mu\nu} = 0_\mu$ holds along solutions $\psiPH$ of the Euler-Lagrange equations.
  One notes however that neither of the above two canonical stresses are symmetric (in their two lower indices), and thus cannot be readily 
identified with an energy-momentum tensor of the theory. 
 To find such an object we need to symmetrize. 
 A computation shows (cf. \cite{HC1946}) that
\beq 
T_{\mu\nu} = 
\Theta_{\mu\nu}+\frac{\hbar c}{64 \pi i} \p^\la \tr \left\{ \psiPHb (\ga_\mu\ga_\nu\ga_\la - \ga_\nu\ga_\mu\ga_\la) \Pi\psiPH - \mbox{adj.} \right\}
\eeq
where
\beq 
\Theta_{\mu\nu} = \half (T_{\mu\nu} + T_{\nu\mu}) =
\frac{\hbar c}{16 \pi i} \tr\left\{ \Pi\psiPHb (\ga_\nu \p_\mu + \ga_\mu\p_\nu)\psiPH - \mbox{adj.} \right\} 
+\frac{\mPH c^2}{8 \pi} \tr \left(\psiPHb \Pi\psiPH\right)\eta_{\mu\nu}.
\eeq
Similarly,
\beq 
T'_{\mu\nu} = 
\Theta'_{\mu\nu}+\frac{\hbar c}{64 \pi i} \p^\la \tr\left\{ \psiPHb (\ga_\nu\ga_\la \ga_\mu - \ga_\mu \ga_\la \ga_\nu ) 
\Pi \psiPH - \mbox{adj.} \right\},
\eeq
with 
\beq 
\Theta'_{\mu\nu} = \half( T'_{\mu\nu} + T'_{\nu\mu}) 
= \frac{\hbar c}{32 \pi i} \tr \left\{ \psiPHb \ga_\nu \p_\mu \psiPH + \psiPHb \ga_\mu \p_\nu \psiPH - \mbox{adj.} \right\}.
\eeq
Either one of $\Theta_{\mu\nu}$ and $\Theta'_{\mu\nu}$ (which by virtue of \refeq{TTprimediff} differ only by a full divergence) can 
now be used as an energy-momentum tensor, since they are symmetric and divergence-free. 
 Upon contracting such an object with a Killing field (generator of continuous isometry) of the domain, one obtains a conserved Noether current.
  Since the domain is the Minkowski space in our case, one thus obtains the familiar ten integral conservation laws of energy, linear momentum, 
angular momentum, and centeroid location in this way, corresponding to time translations, space translations, spatial rotations, and Lorentz 
boosts, respectively.  

More precisely, let $(x^\mu)_{\mu=0}^3$ denote a system of global rectangular coordinates on the Minkowski space $\Rset^{1,3}$ and 
let $\{ \bu_{(\mu)} = \frac{\p}{\p x^\mu} \}$ be the corresponding Lorentz-orthonormal coordinate frame field. 
 Then each of the four vector fields $\bu_{(\mu)}$ is a Killing field for the Minkowski space, thus giving rise to four conserved 
currents $\{\bP_{(\mu)}\}_{\mu=0}^3$ with components
\beq
P_{(\mu)}^\nu := \Theta'^\nu_\la u_{(\mu)}^\la = \Theta'^\nu_\mu
\eeq
and the corresponding time-independent (i.e. constant) integral quantities that are finite for wave functions of sufficiently rapid decay:
\beq\label{def:EnMom}
\cE := \int_{\Rset^3} {\Theta'}^0_0 d^3x,\qquad \cP_k := \int_{\Rset^3} {\Theta'}_k^0 d^3x,\quad k = 1,2,3.
\eeq
In particular, using \refeq{eq:insidechi}, \refeq{eq:insidepsi}, \refeq{eq:potentials}, and the fact that \refeq{MaxEB} holds for both 
pairs $(\be_+,\bb_+)$ and $(\be_-,\bb_-)$, we have that
\bna
8\pi {\Theta'}^0_0 & = & \frac{\hbar c}{2i} \tr \{ \psiPHb \ga^0\p_0 \psiPH - \p_0\psiPHb \ga^0 \psiPH \} \nonumber\\
& = & \frac{\hbar c}{2i} \tr\left( \psi_-^\dag \p_0 \chi_+ + \chi_-^\dag \p_0 \psi_+ + \chi_+^\dag \p_0 \psi_- 
+ \psi_+^\dag \p_0 \chi_- - \mbox{adj.}\right) \nonumber\\
& = & {\mPH c^2} \left[ \be_+\cdot \be_- + \bb_+ \cdot \bb_- \right] - 
{\hbar c } \nabla \cdot \left[ \phi_- \be_+ - \ba_- \times \bb_+ +\phi_+\be_- - \ba_+ \times \bb_- \right].
\label{EnMomEXPL}
\ena
 (Note that the divergence term vanishes upon integration over all space; it may as well be ignored, therefore.)
 When evaluated on a {\em self-dual} wave function, i.e. one for which $\be_+ = \be_-$ and $\bb_+ = \bb_-$, 
the energy density of the photon wave function $\psiPH$
defined in \refeq{EnMomEXPL} resembles the well-known expression for the energy density of 
a formal electromagnetic field $(\be,\bb)$,\footnote{The appearance here of the 
   mathematical expression for the energy density of a classical electromagnetic field is one indication that
   Maxwell's classical electromagnetic field theory may emerge from an appropriate $N$-body generalization of our 
   quantum-mechanical photon theory  in the large $N$ limit through the Law of Large Numbers.}
 and is therefore nonnegative.
 For a general wave function, however, this quantity does not have a sign, in agreement with our understanding of the existence 
of photons with both positive and negative energies.

\section{A quantum probability current for the location of a photon}\label{sec:riesz}

  None of the conserved currents we have discussed so far can play the role of 
the correct generalization of Dirac's quantum probability current $j_{\mbox{\tiny el}}^\mu = \psiELb\ga^\mu \psiEL$ for the 
location of an electron.
  All of them suffer from one or more of the following problems: wrong transformation law 
(i.e. not a Lorentz four vector); or if they transform properly, either lack of pointwise gauge-invariance, 
or lack of positivity for the time component. 

  Interestingly, yet another set of conservation laws for the photon equation exists, which, unlike the Noetherian laws described in the 
previous section, are not derivable from the Lagrangian \refeq{Lagrangian}; indeed, not from any action principle for \refeq{eq:DiracEQphoton}
where variations are to be taken with respect to the wave function itself.\footnote{They can be derived from a ``sub-action principle,'' though,
 indeed well-known from classical electromagnetic field theory, but which can only be used to derive a subset of the full set of equations which
 constitute our photon wave equation. 
 Moreover, to carry out the variations one has to invoke another subset of these equations by fiat.}
 These conservation laws were first discovered by M. Riesz \cite{Rie1946}, in connection with his generalization 
of the massive Dirac equation to rank-two bi-spinors (or {\em Clifford numbers}, as he called them.)
 We will find amongst the conserved Riesz quantities a distinguished four-vector with a positive time component
that, after normalization, qualifies as a quantum probability:
 it plays the desired role of a non-negative probability density ``of finding the light quantum at a given point on a 
given spacelike hypersurface,'' and it transforms in the right manner under the full Lorentz group.

\subsection{The Riesz tensor and its conservation laws}

In a mathematically ingenious work \cite{Rie1946} that remains under-appreciated by the physics community\footnote{See 
  Garding \cite{GardingHISTORY} p. 229,  for an account of the reception of M. Riesz's work on quantum mechanics among 
  his physicist contemporaries.}, 
the great analyst Marcel Riesz argued that, since the Dirac electron wave function --- a rank-one bi-spinor --- belongs to 
a {\em minimal left ideal} of the Clifford algebra $\mbox{Cl}_{1,3}(\Rset)_\Cset$ (see subsection \ref{sec:clifford}), in order to better understand the Dirac equation 
and its many symmetries and conservation laws, one should first {\em extend it to the complete Clifford algebra}, find all the 
conservation laws that one can, and then descend back to the minimal left ideal. 
 He thus advocated for the study of the following Dirac-type equation for a massive, Clifford-algebra-valued section $\psi$:
\beq\label{eq:DiracRiesz}
-i\hbar \ga^\mu \p_\mu \psi + m c \Id \psi = 0,\qquad \psi \in \mbox{Cl}_{1,3}(\Rset)_\Cset \cong M_4(\Cset).
\eeq
 Riesz in particular showed that the rank-two tensor 
\beq\label{RieszTENSORdef}
\bt^{\mu\nu} := (\overline{\psi} \ga^\mu\psi\ga^\nu)_S^{}
\eeq
 (recall that $a_S^{}$ denotes the scalar part of a Clifford number $a\in \mbox{Cl}_{1,3}(\Rset)_\Cset$)  
will be divergence-free along solutions of \refeq{eq:DiracRiesz}, viz.
\beq\label{RieszConservationLaw}
\p_\mu \bt^{\mu\nu} = 0
\eeq 
for any solution $\psi$ of \refeq{eq:DiracRiesz}.
 Riesz further observed that contracting this tensor with the unit timelike vector $X = (1,0,0,0)^T$ yields a conserved 
positive-definite quantity, which he identified with the electron current previously found by Dirac.  

 Our photon wave equation \refeq{eq:DiracEQphoton} differs from Riesz's equation \refeq{eq:DiracRiesz} in one crucial way:
in place of the identity operator in \refeq{eq:DiracRiesz} we have the projection operator $\Pi$.
  It is not hard to see that, since $\Pi$ does not commute with the $\ga^\nu$, the conservation law 
\refeq{RieszConservationLaw} does {\em not} hold for our photon wave equation. 
 Nevertheless, we do obtain a conserved Riesz tensor provided we restrict ourselves to the projected wave function $\Pi\psiPH$, 
i.e. to the {\em diagonal blocks} of $\psiPH$.

 Recall that if $\psiPH$ is a solution of \refeq{eq:DiracEQphoton}, then $\phiPH :=\Pi\psiPH$ solves the massless Dirac equation
\beq \label{masslessDirac}
i\ga^\mu \nab_\mu \phiPH = 0.
\eeq
We now use the well-known ``method of multipliers" to obtain the system of conservation laws
\beq \label{gaugeinvconslaw}
\nab_\mu \left( \overline{\phiPH} \ga^\mu \phiPH \ga^\nu\right) = 0.
\eeq
 To accomplish that, we consider equation \refeq{masslessDirac} and its adjoint
\beq \label{adjointmasslessDirac}
-i \nab_\mu \overline{\phiPH} \ga^\mu = 0.
\eeq
 Let us multiply \refeq{masslessDirac} on the left by $\overline{\phiPH}$ and on the right by $\ga^\nu$. 
 Let us also multiply the adjoint equation \refeq{adjointmasslessDirac} on the right by $\phiPH \ga^\nu$.
  Subtracting the two resulting equations and using that our $\gamma$ matrices are covariantly constant, 
we readily obtain \refeq{gaugeinvconslaw}.

 We are thus motivated to define a two-indexed object $\boldsymbol{\tau}$ with components
\beq \label{def:tau}
\tau_{\mu\nu} := \frac{1}{4}\tr\left( \overline{\phiPH} \ga_\mu \phiPH \ga_\nu \right).
\eeq
 We will call $\btau$ the {\em Riesz tensor} of $\psiPH$, in honor of its discoverer M. Riesz \cite{Rie1946}.  
 We have
\begin{lem}
$\boldsymbol{\tau}$ is a real symmetric Lorentz-covariant four-tensor of rank two, and along solutions of the Dirac equation satisfied by $\psiPH$, 
it satisfies the conservation laws
\beq \label{taudivfree}
\nabla^\mu \tau_{\mu\nu} = 0,\qquad \nu = 0,1,2,3,
\eeq
where $\nabla$ is the covariant derivative on spinors, with respect to which both the metric tensor and the Dirac matrices are constant. 
\end{lem}
\begin{proof}
Tensor properties of $\boldsymbol{\tau}$ are evident from the definition, and the divergence-free property was established in the above. 
 To see that $\boldsymbol{\tau}$ is real and symmetric, we recall that, by choosing a Lorentz-orthonormal frame, one can find vectors 
$\beff_\pm \in \Cset^3$ such that for $\psiPH$ of the from \refeq{defpsiPH}, we have $\psi_+ = i\siV\cdot\beff_+$ and
 $\psi_- =i \siV'\cdot\beff_-^\ast = - i\siV\cdot\beff_-^\ast$. 
 We can then compute all of the components $\tau^{\mu\nu}$ in terms of $\beff_\pm$:
\bna \label{tau00}
\tau^{00} & = & \half\left( |\beff_+|^2 + |\beff_-|^2 \right) \\ \label{tau0j}
\tau^{0j} = \tau^{j0} & = & \frac{1}{2i} \left( \beff_+^\ast \times \beff_+ + \beff_-^\ast \times \beff_- \right)^j \\ \label{taujk}
\tau^{jk} & = & -\mbox{Re} \left( f_+^{j\ast} f_+^k + f_-^{j\ast}f_-^k\right)  + \half (|\beff_+|^2+|\beff_-|^2) \de^{jk},
\ena
which establishes the claim. 
\end{proof}
\begin{rem} \label{tauT}
\textit{
In fact $\boldsymbol{\tau}$ is the sum of two tensors $\boldsymbol{\tau}_\pm$, 
\beq 
\tau_+^{\mu\nu} = \frac{1}{4} \tr\left(\psi_+^\dag {\si'}^\mu \psi_+ \si^\nu\right),\qquad
\tau_-^{\mu\nu} = \frac{1}{4} \tr \left( \psi_-^\dag \si^\mu \psi_- {\si'}^\nu \right),
\eeq
each of which separately satisfies \refeq{taudivfree}.
Furthermore, recalling that $\psi_+ = \Si(\bff_+)$ and $\psi_- = \Si'(\bff_-)$, a computation shows that if we view $\bff_+$ and $\bff_-$ 
(formally) as two electromagnetic Faraday tensors, then the $\boldsymbol{\tau}_\pm$ are nothing but the (Maxwell-type) 
energy-momentum tensors of these two fields:
\beq 
\tau_\pm^{\mu\nu} = \TMmunu[\bff_\pm],
\eeq
where
\beq \label{def:EMenergytensor}
\TMmunu[\bff] := f^{\mu\la} f_{\nu\la} - \frac{1}{4}\eta^{\mu\nu}f_{\al\beta}f^{\al\beta}.
\eeq
 It is known that the electromagnetic energy-momentum tensor $\TMmunu$ satisfies the {\em dominant energy condition}, namely:
 If $X^\nu$ denote the space-time components of the vector field $X$, then
\begin{enumerate}
\item $(\TMmunu X_\nu)$ is causal and future-oriented whenever $X$ is causal and future-oriented.
\item $\TMmunu X_\mu X_\nu \geq 0$ whenever $X$ is causal and future-oriented.
\end{enumerate}
 It follows that the tensors $\boldsymbol{\tau}_\pm$, as well as their sum $\boldsymbol{\tau}$ enjoy the same property, 
and this is crucial for what ensues.
}
\end{rem}

\subsection{The quantum-mechanical probability current for the location of the photon}

 Armed with this non-Noether set of conservation laws, we are now in a position to identify our candidate for the photon probability current. 
 We note that $\tau^{\mu\nu}$ already addresses two of the three difficulties that the previously obtained conserved quantities had,
 namely, it is manifestly gauge invariant, and it enjoys positive-definiteness. 
 One  problem seemingly remains, though: $\tau$ is a rank-two tensor, and therefore in order to construct a Lorentz vector out of it, 
one has to contract it with another vector, which would have to be a Killing vector of the spacetime in order for the resulting 
four-vector to be a conserved current.
 The problem is that Minkowski space has infinitely many Killing vectors (every constant vector field is a Killing vector for the 
Minkowski metric), so to define a {\em unique} current one needs an objective principle which selects a distinguished Killing vector 
field.\footnote{\label{LabFrame}One might be tempted to use,  as a ``distinguished'' Killing vector,
the timelike unit vector of the ``Lab frame'' for which probabilities are computed.
 Indeed, Leopold, whose thesis \cite{Leopold} is based on unpublished notes by Norsen and Tumulka, 
and who thus uses as photon wave equation two independent sets of Maxwell equations pertinent to the two possible chiralities of the
electromagnetic fields, adds the pertinent two energy-momentum tensors in an ad-hoc manner, then contracts them with the timelike 
unit vector of the ``Lab frame.''
 However, there are compelling reasons for why a quantum probability current for a particle should only depend on the wave function 
of the particle, not on some external ``observer.''
 For a critique of theories with observer-dependent probability currents, see Struyve et al. \cite{WardETal}, and 
Tumulka \cite{RodiUnromantic} (and references therein.)}  
 Fortunately, the tensor $\boldsymbol{\tau}$ is associated with its own unique timelike Killing vector field which 
depends only on $\phiPH$. 
 This allows us to construct a distinguished conserved probability current for the photon. 
 
 We anticipate that our probability current for the photon is compatible with {\em Born's rule} for quantum probabilities. 
 To see this one needs a Hilbert space formulation, which we will supply in the last of the main sections. 

 \subsubsection{The timelike Killing field generated by $\phiPH$}

 Let $(x^\mu)$ be a global rectangular coordinate system on the configuration space, which is taken to be a copy of Minkowski space,
and let the Vierbein $\cV := \{\bu_{(\mu)}\}_{\mu = 0}^3$ denote the corresponding global Lorentz-orthonormal frame of constant vector fields 
$\bu_{(\mu)} = \frac{\p}{\p x^\mu}$, with $\bu_{(0)}$ timelike, and the other three spacelike, unit vectors. 
 On Minkowski space, every constant vector field is a Killing field, therefore each one of the $\bu_{(\mu)}$ 
vector fields generates a conserved quantity, as follows:  
 Let 
\beq
\bR_{(\mu)} := \btau \bu_{(\mu)},\qquad \mu = 0,\dots,3
\eeq
(thinking of $\btau$ as a (1,1) tensor, i.e. a linear transformation.) 
Then \refeq{taudivfree} implies that
\beq\label{conserved}
\p_\nu R_{(\mu)}^\nu = \p_\nu (\tau^{\nu}_\la \bu_{(\mu)}^\la) = 0.
\eeq
 Thus we obtain four conserved currents $\{ \bR_{(\mu)}\}_{\mu =0}^3$.  
 Let $\Si_t$ denote the spacelike hypersurface $\{ x^0 = ct\}$.  
 Integrating \refeq{conserved} on $\Si_t$ and applying the divergence theorem, we obtain that, so long as $\phiPH$ is sufficently 
rapidly decaying at spatial infinity, we have 
\beq
\p_t \int_{\Si_t} R_{(\mu)}^0 d^3x = 0.
\eeq
 We thus obtain four constants
\beq\label{def:pimus}
\pi_{(\mu)} := \int_{\Si_0} \tau^{0\mu} d^3 x,\qquad \mu = 0,\dots,3.
\eeq
 Note that the components of $\btau$ are homogeneous of degree two in $\psiPH$, therefore so are the constants $\pi_{(\mu)}$. 
 In particular, using \refeq{tau00} and \refeq{tau0j} we have
\bna\label{piNULL}
\pi_{(0)} & =& \int_{\Si_0} \half\left( |\be_+|^2 + |\bb_+|^2 + |\be_-|^2 + |\bb_-|^2\right) d^3x\\
\pi_{(j)} & = & \int_{\Si_0} \left(\be_+\times \bb_+ + \be_-\times\bb_- \right)^j d^3x.
\label{piJAY}
\ena
Furthermore, as \refeq{def:pimus} makes it clear, these quantities are determined by the {\em initial data} for the wave function, i.e. the values that $\phiPH$ takes on the initial hypersurface $\Si_0$.  

 Let $\bpi$ be the vector whose components in the $\cV$ frame are the $\pi_{(\mu)}$:
\beq
\bpi := \sum_{\mu=0}^3 \pi_{(\mu)} \bu_{(\mu)}.
\eeq
  For finite right-hand sides in \refeq{piNULL}, \refeq{piJAY} 
$\bpi$ is null only if $\be_+\cdot\bb_+=0=\be_-\cdot\bb_-$ almost everywhere; otherwise $\bpi$ is future-directed timelike. 
 Thus $\bpi$ is always future-directed and causal:
\beq
\pi^0 \geq \surd{ \textstyle{\sum\limits_{j=1}^3} (\pi^j)^2}.
\eeq
 Since the necessary condition for $\bpi$ being null is not generically fulfilled, it now follows that $\bpi$  is future-directed and \emph{typically timelike}.\footnote{By
   restricting the initial data for the wave function to typical data, we can ensure that $\bpi$ is always strictly timelike.}

 Suppose that $\cV' := \{\bu'_{(\mu)}\}_{\mu=0}^3$ were another Lorentz-orthonormal frame for the Minkowski configuration spacetime. 
 There must be a Lorentz transformation $\bLa \in O(1,3)$ such that 
\beq
\bu'_{(\mu)} = \bLa \bu_{(\mu)},\qquad \mu = 0,\dots,3.
\eeq
Given that the components of $\btau$ with respect to a given frame transform by similarity: 
$[\btau]_{\cV'} = \bLa [\btau]_{\cV} \bLa^{-1}$ it is easy to see that the components of the vector $\bpi$ transform like those of a 
Lorentz 4-vector, i.e. 
\beq 
[\bpi]_{\cV'} = \bLa [\bpi]_{\cV}.
\eeq
Thus $\bpi$ is a constant vector field on the Minkowski space, and therefore it is a distinguished Killing field of the configuration 
Minkowski space that can be constructed entirely out of the photon wave function $\phiPH$.  
\subsubsection{The probability current}
 We have seen that for finite right-hand sides in \refeq{piNULL}, \refeq{piJAY} the vector $\bpi$ is 
future-oriented and \emph{typically timelike}.
 Thus, except at most for untypical situations, we may assume that the causal vector $\bpi$ is not null,\footnote{If 
   $\bpi$ is null, we can still define a current by setting $\bX = \bpi$. 
 In this case the current $\bj$ will also be null, so the probability density is no longer defined, but one can still have particle 
   trajectories by insisting that the 4-velocity of the particle is parallel to $\bj$ (cf. \cite{BerndlETal}).}
i.e. 
\beq
\boldsymbol{\eta} (\bpi,\bpi) = \eta_{\mu\nu}\pi^\mu \pi^\nu = 
 (\pi^0)^2 - \sum_{i=1}^3 (\pi^i)^2  =: |\bpi|^2 >0.
\eeq
 Set
\beq
\bX := \frac{\bpi}{|\bpi|^2}.
\eeq
 We will use $\bX$ to construct the probability current for a single photon. 
 Let
\beq \label{def:j}
\bMj := \btau \bX.
\eeq
 That is,
\beq
j^0 = \tau^0_\nu X^\nu = \frac{\pi^0 \tau^0_0 + \pi^i \tau^0_i}{|\bpi|^2},\qquad
j^i = \tau^i_\nu X^\nu = \frac{\pi^0 \tau^i_0 + \pi^k \tau^i_k}{|\bpi|^2}.
\eeq
 We propose the four-co-vector field $\bMj^\flat$ to be the \emph{quantum-mechanical probability current for a photon}.\footnote{After 
   learning about our proposal, Ward Struyve informed us (private communication, May 2018) that in his thesis \cite{StruyvePHD}, p.36, 
   he pointed out that the contraction of the energy-momentum-stress tensor of Harish-Chandra's photon wave function with the pertinent conserved
   energy-momentum four-vector produces, after normalization, an observer frame-independent conserved current in Harish-Chandra's formalism, 
   which transforms like a probability current under Lorentz transformations.
   However, while the time component of this current is indeed non-negative, Struyve argues that this current is rather describing the 
   ``energy flow lines'' of a solution to Harish-Chandra's photon wave equation, not a probability current for the position of a photon. 
    Our proposal, while superficially similar, does not encounter such conceptual difficulties because the Riesz tensor is \emph{not} the 
    energy-momentum-stress tensor of our photon wave function.\label{fn:Ward}}
 We note that, since the timelike four-vector $\bX$ is a Killing vector, this is a conserved current, viz.
\beq \label{eq:Jconserv}
\p_\mu j^\mu = 0.
\eeq

 Given a real orthonormal Lorentz frame $\{\bu_{(\mu)}\}_{\mu = 0}^3$, with $\bu_{(0)}=\p_{ct}$ timelike, with
$t$ denoting the time coordinate on $\cM$ in this frame, we now define 
\beq \label{rhoANDjV}
\rho c:= j^0 \geq 0,\qquad \jV := (j^1,j^2,j^3)^{\mathrm{T}}.
\eeq
In this space\ \&\ time splitting, the continuity equation \refeq{eq:Jconserv} for the current $\bMj$ then reads 
\beq \label{eq:PROBcontinuityEQ}
\p_t \rho + \nab \cdot \jV = 0.
\eeq
 Let $\Si_t$ denote the (spacelike) $t$-level-hypersurfaces, and let $\sV$ denote the space vectors in $\Si_t$. 
  The total integral over $\Si_t$ of $\rho(t,\sV)$ is equal to one, which allows it to be viewed as a probability density
for finding a photon at $\sV\in\Si_t$.
  The three-vector field $\jV(t,\sV)$, then, is the corresponding probability current vector-density. 

\begin{rem}
\textit{With view toward the question whether a de Broglie--Bohm (dBB)-type quantum theory for the photon is feasible,
we note that $j^0 \geq |j^i|$, so wherever $j^0\neq 0$ we may define a subluminal three-velocity vector field $\vV(t,\sV)$ componentwise by
\beq 
 v^i := \frac{j^i}{j^0}.
\eeq
  In this space\ \&\ time splitting, the continuity equation \refeq{eq:Jconserv} for the current $\bMj$ reads:  
\beq \label{eq:PROBcontinuityEQsplit}
\p_t \rho + \nab \cdot ( \rho \vV) = 0.
\eeq
  The space vector field $\vV$ can then be given the meaning of a {\em guiding velocity field} for the position of a photon.
  Note that in this interpretation one {\em must} abandon the notion that $\rho$ is fundamentally a probability density --- 
  it {\em only serves in this role for all practical purposes}.
 Note also that a photon would move at subluminal speeds in this dBB-type theory; however, an essentially ``freely moving'' photon
would eventually move at a speed arbitrarily close to the speed of light because the wave function will locally become a plane
wave.}
\end{rem}
%
\section{Hilbert space formalism: Photon Hamiltonian and Born's rule}\label{sec:HilbertBorn}\vspace{-5pt}
%
\subsection{The photon wave equation as a Schr\"odinger equation}\label{sec:schrodA}\vspace{-5pt}
 The obvious procedure to arrive at the Schr\"odinger -- or Hamiltonian -- form of the Dirac-type equation
\refeq{eq:DiracEQphotonAGAIN} would seem to consist of choosing a Lorentz frame, equivalently a space-and-time splitting, 
multiplying \refeq{eq:DiracEQphotonAGAIN} by $\ga^0$, and moving all terms other than the one involving the time 
derivative of $\psiPH$ to the right-hand-side. 
 This would yield
\beq\label{eq:SCHROEDfull}
i\hbar \p_{t} \psiPH = 
\tilde{H}\psiPH,
\eeq
with Dirac-type operator
\beq\label{eq:Hamfull}
\tilde{H} := -i\hbar c \ga^0 \ga^k \p_k + \mPH c \ga^0 \Pi.
\eeq
 However, taking this route one immediately encounters an obstacle: 
 It is not clear what the Hilbert space structure is, with respect to which 
$\tilde{H}$
would be at least formally self-adjoint. 
 The first term in 
$\tilde{H}$ suggests that the relevant quadratic form is $ \tr( \psiPH^\dag \psiPH )$, 
but this expression mixes the $+$ components in $\psiPH$ with the $-$ components, and the gauge-invariant quantities with gauge-dependent ones. 
This and other conceptual difficulties make it hard to make quick progress in this direction.\vspace{-5pt}

%
\subsection{The diagonal part of the photon wave equation as a Schr\"odinger equation}\label{sec:schrodB}\vspace{-5pt}
 %
Fortunately, the conserved probability current for our relativistic photon wave equation \refeq{eq:DiracEQphotonAGAIN}, 
which we identified with the help of the M. Riesz tensor, also suggests a Hamiltonian and a Hilbert space 
structure to work with.
 Namely, the Riesz tensor is constructed from the diagonal part $\Pi\psiPH =: \phiPH$ of $\psiPH$, which
as we noted earlier, satisfies the massless Dirac equation, i.e. \refeq{eq:DIRACmZEROeqn}.
 In space \&\ time splitting the mass-less Dirac equation \refeq{eq:DIRACmZEROeqn} for $\phiPH$ then
takes the Schr\"odinger-type format
\beq 
\label{DIRACeqSCHROEDINGERformat}
i\hbar \p_t \phiPH = H\phiPH,
\eeq
with formal Hamiltonian 
\beq \label{ham}
H 
:=  -i\hbar c \al^k \p_k,
\eeq
where $\al^k := \ga^0\ga^k$.
 It is this Schr\"odinger equation, \refeq{DIRACeqSCHROEDINGERformat}, 
and Hamiltonian, \refeq{ham}, which allows us to incorporate the usual quantum-mechanical formalism.
\begin{rem}
\textit{There is no corresponding Schr\"odinger-type equation for $(\Id - \Pi)\psiPH$, because the off-diagonal components of $\psiPH$ 
satisfy the non-autonomous equations \refeq{DIRACeqSPLIToffDIAG}.}
\end{rem}

\subsubsection{The diagonal Hamiltonian and the diagonal  Hilbert space structure}
 Let $\cT_0$ be the sub-bundle of  $\cT$ consisting of {\em block-diagonal} rank-two bi-spinors, i.e. those of the form
\beq
\phiPH = \left(\begin{array}{cc} \psi_+ & 0 \\ 0 & \psi_- \end{array}\right),\qquad \tr \psi_+ = \tr \psi_- = 0,
\eeq
defined on the 1-particle configuration (Minkowski) spacetime $\cM$. 
 Let $(x^\mu)_{\mu=0}^3$ be a rectangular coordinate system on $\cM$, and let $\Si_t$ denote the spacelike hyperplane $x^0 = ct$ 
in $\cM$, which is isometric to $\Rset^3$ with the Euclidean metric.
 For $\phi \in \cT_0$ we denote by $\phi_t$ the restriction of $\phi$ to $\Si_t$, thought of as a bispinor field on $\Rset^3$.
   For $\phi,\ze \in \cT_0$ we define their inner product at time $t$ to be
\beq\label{def:innerprod}
\langle \phi_t | \ze_t \rangle := \int_{\Si_t} \tr\left( \phi_t^\dag \ze_t \right) d^3x,
\eeq
which is easily seen to satisfy the requirements of being a complex-valued inner product. 
 We define the Hilbert space $\cH$ to be the completion of smooth, compactly supported bi-spinor fields defined on $\Rset^3$, 
with values in $\cT_0$, under the $L^2$ norm induced by this inner product, viz.
\beq
\|\phi_t\|_{L^2}^2 : = \langle\phi_t | \phi_t\rangle = \int_{\Rset^3} \tr\left( \phi_t^\dag\phi_t \right) d^3x.
\eeq

 The diagonal Hamiltonian $H = -i\hbar c\al^k\p_k$, written more explicitly as
\beq 
\label{DIRACphotonHAMILTONIAN}
H = -i\hbar c\begin{pmatrix} \siV\cdot\nab  & 0 \\
                                        0 & - \siV\cdot\nab
\end{pmatrix},
\eeq
is easily seen to be symmetric (i.e. formally self-adjoint) with respect to this Hilbert space inner product, i.e.
\beq 
\langle \phi | H\zeta \rangle = \langle H\phi| \zeta \rangle,
\eeq
for bi-spinors $\phi$ and $\zeta$ in $C^\infty_c(\Rset^3;\cT_0)$.
 Furthermore, by following verbatim the well-known proof --- using Fourier 
transform --- for the essential self-adjointness of the electronic free Dirac Hamiltonian (see, e.g., Thaller \cite{ThallerBOOK}, \S1.4.4),
our diagonal $H$ is easily seen to be essentially self-adjoint.

\subsubsection{The Einstein relations}\vspace{-5pt}

 The spectrum of the self-adjoint Hamiltonian $H$ given in \refeq{DIRACphotonHAMILTONIAN} is ${\mathrm{spec}\,}H=\Rset$.
 Writing the generalized eigenvalues as
\beq 
\label{HAMILTONIANspec}
E=\pm \hbar c|\kV|, \qquad \kV\in\Rset^3,
\eeq
and recalling dispersion relation \refeq{dispREL}, we recover Einstein's energy-frequency relation for photons,\footnote{Einstein 
  in his first statistical analysis of Planck's 
  black body law wrote  the equivalent formula $E=R\beta\nu/N$ (p.143 of \cite{EinsteinPHOTONa}), with $R$ the ideal gas constant and 
  $N$ Avogadro's number; setting $\beta=h/k_{\mbox{\tiny{B}}}$, where $k_{\mbox{\tiny{B}}}$ is Boltzmann's constant, this is formally 
  identical with Planck's relation $E=h\nu$ for the energy unit that can be absorbed from / emitted into an 
  electromagnetic wave of frequency $\nu$ by ``atomic oscillators'' in the material walls of a black body radiator.
    In his commentary on the works of Lorentz, Jeans, and Ritz (\cite{EinsteinPHOTONb}, p.191) and his 
  subsequent statistical analysis of Planck's black body formula (\cite{EinsteinPHOTONbb}, p.497) Einstein used $E=h\nu$. 
    Incidentally, to the best of our knowledge, Planck never accepted Einstein's ``Lichtquantenhypothese,'' the ``hypothesis of 
  [the existence of] light quanta'' (a.k.a. photons).
    Thus, \emph{in the context of the theory of photons} it seems appropriate to refer to $E=\hbar\omega$ as 
 ``Einstein's energy-frequency relation of a photon,'' and not as a ``Planck(--Einstein) relation.''}
\beq 
\label{EinsteinEomega}
E=\hbar \omega,
\eeq
here with positive and negative frequencies $\omega\in\Rset$.
 Since the momentum operator $-i\hbar\nab$ commutes with $H$, we also find the generalized momentum eigenvalue vectors to
satisfy Einstein's momentum - wave vector relation for a photon,\footnote{Cf. \cite{EinsteinPHOTONbb}, p.497, 
and \cite{EinsteinPHOTONc},  p.61; there he writes $h\nu/c$ for the magnitude of momentum transferred between a photon and 
 a molecule upon the photon's directed emission / absorption. 
   This relation is formally identical with de Broglie's formula
 $\hbar\kV =\pV$ for the wave vector of a phase wave associated with a massive electromagnetic particle such as the electron 
 having  momentum $\pV$.
   For this reason this momentum-wave vector relation is often referred to as the ``de Broglie relation.''
   However, to the best of our knowledge, de Broglie proposed this relation  in 1924, for massive particles. 
   Thus, \emph{in the context of a theory of photons} we find it appropriate to refer to ``$\pV=\hbar\kV$'' 
 as ``Einstein's momentum-wave vector relation for the photon.''
   Also Compton in his 1923 electron--X-ray scattering theory  \cite{Compton} 
 used $p=h\nu/c$ for the magnitude of the momentum of photons as if it was a generally known fact.}
\beq 
\label{EinsteinPk}
\pV=\hbar\kV.
\eeq

\subsection{Born's rule for the quantum probabilities}\label{sec:Bornrule}\vspace{-5pt}

Given $\psiPH$, the photon probability current \refeq{def:j} 
produces the probabilitistic quantum formalism usually summarized as ``Born's rule,'' so long as the vector $\bpi$ is not null.
 Explicitly, pick any orthonormal Lorentz frame $\{\bu_{(\mu)}\}_{\mu = 0}^3$, with $\bu_{(0)}=\p_{ct}$ timelike, 
and express $\bMj$ in terms of $\psiPH$, recalling \refeq{def:tau}.
 This yields a bilinear expression for $\bMj$ in terms of the normalized $\phiPH$. 
  In particular,
\beq
\rho  =  \frac{\pi^0}{|\bpi|^2} \frac{|\be_+|^2 + |\bb_+|^2 + |\be_-|^2 + |\bb_-|^2}{2}  -
 \sum_{i=1}^3 \frac{\pi^i}{|\bpi|^2} \left( \be_+\times \bb_+ + \be_-\times\bb_- \right)^i
\eeq
and thus, for all $t$,
\beq
\int_{\Si_t} \rho\  d^3x = 1.
\eeq
 If the vector $\bX$ is not null, one can always find a Lorentz transformation $\bLa$ such that 
\beq
\bLa \bX = \left(1, 0, 0, 0\right)^{\mathrm{T}}.
\eeq 
 In this co-moving frame the expression for the probability density simplifies to
\beq
\rho = \frac{1}{\pi^0} \frac{|\be_+|^2 + |\bb_+|^2 + |\be_-|^2 + |\bb_-|^2}{2}
 =
 \frac{ \tr\left ({\phi^\dagger_{\mbox{\rm\tiny ph}}}\phiPH\right)}{\int_{\Si_0} \tr\left ({\phi^\dagger_{\mbox{\rm\tiny ph}}}\phiPH\right) d^3x} =
 \frac{\tr\left ({\phi^\dagger_{\mbox{\rm\tiny ph}}}\phiPH\right)}{\langle \phiPH | \phiPH \rangle}.
\eeq
 This is manifestly compatible with Born's rule.
 More to the point, for normalized diagonal $\phiPH$, viz. when choosing the conserved quantity $\langle\phiPH | \phiPH\rangle=1$, 
this becomes $\rho = |\Pi\psiPH|^2$.\vspace{-10pt}

\section{Summary and Outlook}\label{sec:conclusions}\vspace{-5pt}

  Beginning with Majorana \cite{MajoranaUnpublished} and Oppenheimer \cite{OppiPHOTON} the problem of finding a
Lorentz-covariant \emph{quantum-mechanical} wave equation for a single photon has occupied many minds; see e.g. Kobe \cite{Kob1999} 
and references therein.
 In particular, several authors, e.g. \cite{OppiPHOTON, Good1957, SachsSchwebel, MosesQM, BiBiTHREE, TamburiniVicino, Mohr} have looked 
for a wave equation of the photon structurally similar to the Dirac equation for the electron; see also
{\S}12 in \cite{IBBphotonREV}.
 However, none of these proposed ``photon wave equations'' derive from a Lorentz-scalar action principle, and none
leads to an acceptable probability current for finding the photon at a particular point in space  which realizes the
usual $L^2$ Hilbert space ``Born's rule'' in a space-plus-time split of spacetime.
 It has even been questioned whether a photon analog of Dirac's wave equation for an electron exists. 

 In the main part of this paper we have shown that the Dirac-type equation \refeq{eq:DiracEQphoton} (i.e. \refeq{eq:DiracEQphotonAGAIN}) for 
the trace-free rank-two bi-spinor field $\psiPH$ is a viable candidate for the quantum-mechanical wave equation of the photon.
 We note that we operate with the maximal number of components of a photon wave function as required by representation theory,
in contrast to all previous proposals which we are aware of.
 We believe that our proposal, in its entirety, is new and that it overcomes all the obstacles encountered in the previous proposals
by other authors. 

 In particular, our photon wave equation \refeq{eq:DiracEQphoton} derives from a Lorentz-scalar Lagrangian, as we have shown 
in subsection \ref{sec:action}, and it furnishes a conserved probability current for ``finding the photon at a particular location,'' 
based on the non-Noetherian conservation laws associated with the M. Riesz tensor, which transforms properly under the full Lorentz group.

 We remark that our special-relativistic photon wave equation, furnishing a distinguished probability current four-vector density as it
does, could seem to be in violation of the so-called\footnote{As Weinberg and Witten point
  out, the part of the theorem which concerns conserved four-vector current densities, is due to Sidney Coleman.}
 Weinberg--Witten theorem \cite{WeinbergWitten}.
 However, the quantum-field theoretical assumptions made in that theorem do not cover the quantum-mechanical formulation used in 
the present paper. 

 Our probability current in turn suggests the usual $L^2$ Hilbert space structure be used for the Hamiltonian obtained from the
autonomous diagonal part of our photon wave equation, viz. from the
\emph{massless Dirac equation $\ga^\mu p_\mu \phiPH = 0$ for the block-diagonal rank-two bi-spinor wave function} $\phiPH:= \Pi\psiPH$,
rewritten as a Schr\"odinger equation.
 The Hamiltonian has all the features expected of a photon Hamiltonian by relativistic quantum mechanics.
 Most importantly, the four-vector current density that satisfies the continuity equation has a non-negative time component,
which in a co-moving frame (whenever such exists)\footnote{Since, as we explained,
  the probability current four vector is \emph{typically timelike}, not null, by restricting the initial data for the 
  wave function to be ``typical'' we can always ensure that  the probability current four-vector is \emph{timelike}.}
becomes $\frac{\tr({\phi^\dagger_{\mbox{\rm\tiny ph}}} \phiPH)}{\int \tr({\phi^\dagger_{\mbox{\rm\tiny ph}}} \phiPH) d^3x}$, 
with $\int \tr({\phi^\dagger_{\mbox{\rm\tiny ph}}} \phiPH) d^3x$ conserved in time. 
 For normalized wave functions this reads $|\Pi\psiPH|^2$ (where the trace is understood as part of the absolute bar operation).
 Thus it qualifies for the Born-rule of the probability density of finding a photon at a particular point in space,
\emph{with respect to the usual Lebesgue measure}.

 It should be noted that, even though the off-diagonal blocks of $\psiPH$ do not enter our photon particle current, 
they play an integral role in the theory, since without them we would not have been able to formulate a theory that 
is derivable from a relativistically scalar Lagrangian, which would have prevented us from obtaining Noetherian 
conservation laws of energy, momentum, etc. for our photon wave equation. 
 Thus we could not have taken the photon wave function to consist of only the diagonal blocks. 

 Interestingly enough, the coupling between the diagonal and the off-diagonal blocks of the photon wave function
in our Dirac-type equation \refeq{eq:DiracEQphotonAGAIN} introduces a yet-to-be-determined mass parameter, $\mPH$, 
even though any solution $\psiPH$ of \refeq{eq:DiracEQphotonAGAIN} satisfies the massless Klein--Gordon equation. 
 The parameter features prominently in the law of energy conservation of $\psiPH$, see Section \ref{sec:energymom}. 
 Thus we expect that it will be determined through an interacting theory of photons and electrons in which the
off-diagonal blocks should play an important role.

 In this vein, we plan next to take steps in the direction of founding a genuine $N$-body quantum mechanics of 
photons and electrons.
 In particular, we hope to set up a fully-covariant two-body problem involving an electron and a photon, and to study
the Compton scattering process, as well as the emission / absorption of photons by Hydrogen (with a proton treated in 
Born--Oppenheimer approximation as given).
 As a preliminary step, jointly with Matthias Lienert we have already begun a study of the Compton scattering process in 
one space dimension,\footnote{While Maxwell's electromagnetic field equations in 1 space dimension yield only
  trivial solutions, our photon wave equation has non-trivial solutions, indeed.
  This fact alone makes it plain that our photon wave equation is not equivalent to any of the previous proposals which are 
   equivalent to Maxwell's field equations.}
using the multi-time wave function formalism.
  These results will be reported elsewhere.

 All our efforts are meant to shed a rigorous light on (the early) Einstein's idea that ``a quantum of light is being guided by a 
ghost field.''\footnote{Born \cite{BornsPSISQUAREpapersA,BornsPSISQUAREpapersB} reports that Einstein spoke of the guiding
 field as a ``Gespensterfeld'' because, unlike Maxwell's electromagnetic fields, a guiding field would not in itself carry 
 physical energy or momentum, which in turn Einstein thought of as being localized in the light quanta (photons).} 
 Wigner \cite{WignerCollectedWorks} reports that Einstein never figured out how to make his idea into a theory because he 
assumed that each particle's guiding field satisfies its own \emph{independent} ``ghost field equations.'' 
 Wigner further states that this road block was overcome, in the non-relativistic realm, by the realization that the 
$N$-particle configuration is guided by the solution to Schr\"odinger's wave equation on configuration space --- 
an idea due to\footnote{\label{fn:BdBB}How exactly the guiding is being done didn't concern Born, yet he stated that he 
was convinced that the guiding cannot be deterministic; a well-known non-deterministic guiding equation was later supplied 
through Nelson's stochastic mechanics \cite{NelsonA,NelsonB}.
  By contrast, the de Broglie--Bohm guiding law is deterministic, and in a sense the determinisitc limit of Nelson's stochastic 
formalism.}
Born \cite{BornsPSISQUAREpapersA,BornsPSISQUAREpapersB}, de Broglie \cite{deBroglieSOLVAY}, and later Bohm \cite{Bohm52}, 
which Wigner mistakenly attributes to Schr\"odinger. 
 Inspired by the de Broglie--Bohm theory of non-relativistic particle motion as developed in \cite{DGZ}, we expect 
that Einstein's ``ghost field for a photon'' is the conditional wave function of the photon, obtained from
the joint $N$-body wave function of all electrons and photons, where the conditioning is on the actual positions of all the other particles.
 
 In this paper we have only explored the conservation laws for the full photon wave equation \refeq{eq:DiracEQphoton}, which
  follows from an action principle, using Noether's symmetry techniques; and those for the autonomous diagonal part of the 
equation, namely \refeq{eq:DIRACmZEROeqn}, which follow from M. Riesz non-Noether symmetry techniques. 
For the latter equation, many more conservation laws exist.
  Indeed, as already observed by Simulik \cite{Simulik}, any one of the domain transformations generated by the conformal 
Killing vector fields of the Minkowski space $\Rset^{1,3}$ (which form a 15-dimensional Lie group) can be applied concurrently 
with any one of the 8 ``vertical" transformations 
(i.e. along the fiber) that act on a $\phiPH$ satisfying \refeq{eq:DIRACmZEROeqn},
leading to the existence of 128 conserved quantities (although 62 of these turn out to be trivial.)
 Simulik showed that these quantities are in one-to-one correspondence with the conserved quantities for the Maxwell 
system\footnote{In Simulik's work, the 
 Dirac equation \refeq{eq:DIRACmZEROeqn} is not considered as diagonal part of a photon wave equation, but as equivalent to
Maxwell's equations for the classical electromagnetic field in vacuum; see also \cite{Moses1957,Moses1959}
 As to conservation laws for Maxwell's equations, see also \cite{BesselHagen}, \cite{Fushchych}, \cite{AncoPohj}, \cite{AncoThe}.\vspace{-10pt}}
 \refeq{MaxEB} studied previously by Lipkin \cite{Lipkin}, who found `nothing but ``zilch'' (as he called it!),' a complex
$4\times4\times4$ spacetime tensor.
 Some of these additional conserved quantities may turn out to be of relevance to photons as well, and we hope to explore that in the future. 
\bigskip

\noindent
{\bf Acknowledgements:} We thank Shelly Goldstein, Herbert Spohn, and especially Rodi Tumulka for illuminating discussions.  
  After submission of our paper we received many valuable comments on the circulated preprint, for which we thank:
 Olivier Darrigol for his commentary concerning the ``historical origin of $\bE+i\bB$'' (see Appendix \ref{app:others});
 Michel Jansen for pointing us to Darrigol's historical studies; 
 Max Lein and Gian-Michele Graf for their remarks about the works, respectively, of Wigner \cite{Wig1939} and Bargmann--Wigner 
\cite{BarWig1948}, which prompted us to insert clarifying comments regarding representation-theoretical approaches to the photon 
wave equation; Nikolai Leopold for sending us his master thesis \cite{Leopold}, see footnote \ref{LabFrame};
 Ward Struyve for sending us his Ph.D. thesis \cite{StruyvePHD}, see footnotes \ref{fn:mzapp} and \ref{fn:Ward}; 
 Friedrich Hehl for pointing out reference \cite{Cor1953}. 
 We also thank the referee for the helpful suggestions. 
\newpage

\section*{Appendix}\vspace{-5pt}
\appendix

 In the main part of our paper we hope to have put to rest the misperception that a quantum-mechanical wave equation for
a single photon is an impossibility.
 In Appendix \ref{app:others} we comment on some other widespread misperceptions in the opposite direction:
 that the quantum-mechanical photon wave equation is already given by Maxwell's equations for the classical electromagnetic field
--- in disguise ---, and that this could be easily reverse-engineered from QED. 
 In Appendix \ref{app:primer} we provide a brief recap of the algebra of spinors and bi-spinors of ranks one and two, 
 for the convenience of the reader. \vspace{-10pt}

\section{Maxwell's equations, QED, and all that}\label{app:others}

\subsection{Maxwell's equations in complex form: a ``Mock Photon Wave Equation''}\vspace{-5pt}

 Maxwell's equations for the electromagnetic vacuum fields comprise the evolution equations
\begin{alignat}{1}
\textstyle
\pddt{\bB(t,\sV)}
&= \label{eq:MLdotB}
        -c\nab\times\bE(t,\sV) \, ,
\\
\textstyle
         \pddt{\bE(t,\sV)}
&= \label{eq:MLdotE}
        +c\nab\times\bB(t,\sV)  ,
\end{alignat}
and the constraint equations
\begin{alignat}{1}
\textstyle
        \nab\cdot \bB(t,\sV)  
&= \label{eq:MLdivB}
        0\, ,
\\
        \nab\cdot\bE(t,\sV)  
&=\label{eq:MLdivE}
       0\,;
\end{alignat}
here, 
$\sV\in\Rset^3$ and $t\in\Rset$ refer to points in ``space'' and ``time'' as per choice of a Lorentz frame of physical Minkowski 
spacetime, and similarly $\bE$ and $\bB$ are the electric field and the magnetic induction field defined w.r.t. this 
Lorentz frame; see \cite{JacksonBOOK}.

 The two equations (\ref{eq:MLdotB}) and (\ref{eq:MLdotE}) can be merged into a single evolution equation 
for the complex vector field $\bF := \bE +i\bB$,
\begin{equation}
\textstyle
 i \pddt{\bF(t,\sV)}
= \label{eq:MLdotF}
        c\nab\times\bF(t,\sV) \, ,
\end{equation}
and the constraint equations can be merged into
\begin{equation}
\textstyle
        \nab\cdot \bF(t,\sV)  
= \label{eq:MLdivF}
        0\, .
\end{equation}
 This rewriting of Maxwell's equations was presumably known in the late 19th century.\footnote{The complex 
   vector $\bF = \bE +i\bB$ features in ${\S138}$, eq.(2) on p.348, of H. Weber's book \cite{Weber} on ``partial differential 
   equations of mathematical physics based on Riemann's lecture notes.'' 
  Silberstein in 1907 \cite{SilberA,SilberB}, citing Weber's book as source, refers to Riemann as discoverer of this formula.
  Eventually the Bia{\l}ynicki-Birulas \cite{IBBZBBphotonREV} sanctioned $\bF$ the Riemann--Silberstein vector.
  However, Weber subjected Riemann's lecture notes to extensive editing and updating! 
  A 1901 review (by ``G.H.B.'') in Nature \cite{NatureBOOKreview} of Weber's book states 
    ``THE lectures, delivered at the University of G\"ottingen by Prof. Bernard Riemann in the sessions of 1854-55, of 
        1860-61, and in the summer of 1862, have thanks to the volume brought out after Riemann's death under the editorship 
        of Karl Hattendorff, long ranked among the mathematical classics. 
      The third and last edition of ``Partielle Differentialgleichungen'' appeared in 1882, and two years ago Prof. Heinrich 
        Weber was entrusted with the task of bringing out a fourth edition. 
      There were three possible ways in which this task could have been fulfilled. 
      One way was to re-publish the edition of 1882, with trifling additions and alterations. 
      The second way was to retain the existing text, but to add copious notes together with references to recent developments 
        bordering on the subject of Riemann’s lectures. 
      The third way was to write an entirely new book, based, indeed, on the earlier editions, but completely brought up to date 
        by the embodiment of the new methods and problems that have come into existence in connection with discoveries in mathematics 
        and physics extending over nearly twenty years from the date of the last edition, and nearly forty years from the time when 
        the lectures were given by Riemann. 
      Prof. Weber has adopted the last of these alternatives, and by so doing has produced a treatise which will be invaluable to the 
        modern mathematical physicist.''
  Since Maxwell's dynamical theory of the electromagnetic field was published in 1865, Maxwell's equations 
    were certainly not covered by Riemann in his lectures 3+ years before this event. 
  Of course, it is not impossible that Riemann in the last year of his life worked on Maxwell's equations and 
    that he did note the possibility of merging the evolution equations 
    for $\bE$ and $\bB$ into a single complex equation for $\bF = \bE +i\bB$,
 but this is pure speculation.
  To us it seems much more likely that Weber, a skillful mathematician in his own right, introduced $\bF = \bE +i\bB$ himself, 
    presumably sometime in the late 19th century already.
  Thus, ``Weber vector'' would seem a more appropriate name than ``Riemann--Silberstein vector.''
  Perhaps a historian of science will be able to clarify this issue.
 (Note added: We are very grateful to Olivier Darrigol for informing us in correspondences on Jan. 30 \&\ 31, 2018 that, 
  to the best of his knowledge, the earliest known occurrence of the complex field vector $\bF = \bE +i\bB$ is in Heinrich Weber's text
 of 1901, and that Weber does not refer to any earlier source.)}

 It could of course not have been before the advent in 1926 of Schr\"odinger's wave equation, for
anybody to notice that (\ref{eq:MLdotF}), multiplied by $\hbar$, has a close semblance
to what could perhaps qualify as the wave equation for the photon, viz.\footnote{Since $\bF$ has the physical dimension 
of an electric field strength, ``charge/length$^2$,'' while $|\PsiPh|^2$ has the physical dimension of a number density, ``1/length$^3$,''
 one certainly could not just make the identification ``$\bF=\PsiPh$'' --- yet, the linearity of the equations could seem to
suggest that a simple conversion factor would suffice to \emph{formally} map any electromagnetic solution
$\bF$ of \refeq{eq:MLdotF}, \refeq{eq:MLdivF} into a solution $\PsiPh$ of \refeq{eq:MLdotPSI}, \refeq{eq:MLdivPSI}
with the right physical dimensions.\vspace{-10pt}}
\begin{equation}
 i \hbar \textstyle{\pddt}{\PsiPh}(t,\sV)
= \label{eq:MLdotPSI}
        \hbar c\nab\times{\PsiPh}(t,\sV) \, ,
\end{equation}
with the vector wave function ${\PsiPh}$ constrained by
\begin{equation}
\textstyle
        \nab\cdot {\PsiPh}(t,\sV)  
= \label{eq:MLdivPSI}
        0\, .
\end{equation}
 The scare quotes are intended to prevent the reader from confusing ${\PsiPh}(t,\sV)$ with the true quantum-mechanical wave function 
of the photon $\psiPH(t,\bq_{\mbox{\tiny{ph}}})$.

 Ignoring for a moment the vital issue of Lorentz covariance, equation \refeq{eq:MLdotPSI} (constrained by \refeq{eq:MLdivPSI}) indeed
seems to offer everything one could expect from a quantum-mechanical wave equation:
\begin{quotation}
\noindent
 (a)  According to quantum mechanics, if $\PsiPh$ is indeed the photon wave function, then by ``Born's rule'' 
$|\PsiPh|^2  := \PsiPh^*\cdot\PsiPh$ should be the probability density for finding a photon at $\sV$. 
 Indeed, $|\PsiPh|^2$ satisfies the continuity equation
\begin{equation}
\pddt |\PsiPh|^2 + \nabla\cdot c\,\Im\fm\, (\PsiPh^*\times\PsiPh) = 0,
\end{equation}
so that the Hilbert space norm $\|\PsiPh\|^2:= \int_{\Rset^3}\PsiPh^*\cdot\PsiPh d^3s$ is preserved in time; this seems to support
a Born's probability interpretation of $|\PsiPh|^2$.
\smallskip

\noindent
 (b) The formal Hamiltonian given by r.h.s.\refeq{eq:MLdotPSI} has generalized eigenvalues $E = \pm \hbar c|\kV|$; so Einstein's 
relation for photons holds per the dispersion relation $\omega^2 = c^2|\kV|^2$.
\smallskip

\noindent
 (c) $\PsiPh$ transforms like a vector under space rotations, so that it could seem to describe a spin-1 particle.
    Better yet, \refeq{eq:MLdotPSI} can be rewritten in a Weyl-like format 
\begin{equation}
 i \hbar \textstyle{\pddt}{\PsiPh}(t,\sV)
= \label{eq:MLdotPSIinWEYLshape}
         c \widehat{\Sigma}\cdot (-i\hbar\nab){\PsiPh}(t,\sV) \, ,
\end{equation}
where $\widehat{\Si}$ is the vector of Oppenheimer's $3\times3$ spin matrices 
\begin{equation}
\Sigma_x = 
\begin{pmatrix}  0 &  0 & 0 \\ 
                 0 &  0 & -i\\ 
                 0 &  i & 0 
\end{pmatrix},\quad 
\Sigma_y = 
\begin{pmatrix}  0 &  0 & i\\ 
                 0 &  0 & 0 \\ 
                -i &  0 & 0 
\end{pmatrix},\quad
\Sigma_z = 
\begin{pmatrix} 0 & -i & 0 \\ 
                i &  0 & 0 \\ 
                0 &  0 & 0 
\end{pmatrix};
\end{equation}
note that $[\Sigma_x,\Sigma_y]=i\Sigma_z$ (cyclical), and that each $\Sigma_\bullet^{}$ is Hermitian with spec $(\Sigma_\bullet^{})=\{-1,0,1\}$. 
 Thus, $\widehat{S}:=\hbar\widehat{\Sigma}$ yields a ``mock photon spin operator'' with eigenvalues $\pm\hbar$ and $0$ --- with 
the constraint equation \refeq{eq:MLdivPSI} removing the $0$ mode, i.e. absence of longitudinal components, as expected for the photon 
wave function.
\end{quotation}

\noindent
 All this looks very enticing, cf. \cite{SmithRaymer}, \cite{Chandrasekar}, and it also prompted 
a (very preliminary) attempt at a de Broglie--Bohm type theory for a photon, see \cite{Esposito}.

 Of course, we so far have put aside the important question of Lorentz covariance, yet if one now recalls that $\PsiPh$ is essentially
$\bF$ (except for a constant conversion factor), and that $\bF$ is simply a complex linear combination of the Maxwell fields $\bE$ and $\bB$
which satisfy the Lorentz covariant vacuum Maxwell equations, then it could seem that one has accomplished the feat of establishing
the quantum mechanical wave equation for  a single photon. 
  However, if this were the whole story it would be in all textbooks on quantum mechanics ever written, presented jointly with Schr\"odinger's,
Pauli's,  and Dirac's wave equations for electron(s). 
 Alas, things are not that simple, as noted already by Oppenheimer \cite{OppiPHOTON} and more recently by \cite{Cugnon}.

 Namely, even if \refeq{eq:MLdotPSI}--\refeq{eq:MLdivPSI} was covariant under the full Lorentz group, as is
the real Maxwell system \refeq{eq:MLdotB}--\refeq{eq:MLdivB}, the conserved non-negative quantity 
$|\PsiPh|^2 = (\Re\fe\PsiPh)^2 + (\Im\fm\PsiPh)^2$ would transform like the time-time component of a rank 2 tensor --- indeed,
$|\PsiPh|^2 \propto |\bE|^2 + |\bB|^2$, with a constant conversion factor accounting for the proportionality, and the term at
the r.h.s. is the familiar expression of the electromagnetic field energy density. 
 As such it can't play the role of a probability density, 
which is precisely Bohm's, Oppenheimer's, Landau's, and more recently Weinberg's, critiscism quoted in the introduction.
 This has contributed to the (premature) claims that it were \emph{impossible} to construct a probability 
density, for finding a photon at a given space point, which transforms properly under Lorentz transformations.
 It's not impossible, as we have shown.
 Yet, $|\PsiPh|^2$ is certainly not the correct formula for this probability density;
as Weinberg wrote \cite{WeinbergTALK}, p.2.: ``Certainly the Maxwell field is not the wave function of the photon,...'' 

\subsection{Doubling up, plus self-duality: Just twice the same}\label{sec:selfDUAL}

 In close analogy to what we have remarked in section 2, there is a significant difference between these real and the
complex versions of Maxwell's equations. 
 With the convention that $\bE$ transforms like a polar vector and $\bB$ like an axial vector under space inversions, 
the system \refeq{eq:MLdotB}--\refeq{eq:MLdivE} is covariant under the full Lorentz group, while the system 
\refeq{eq:MLdotF}--\refeq{eq:MLdivF} is not; namely, under a space inversion,
$\bF\mapsto -\bF^\ast$, $\nab\times \mapsto -\nab\times$, and $\p_t\mapsto\p_t$, and so \refeq{eq:MLdotF} turns into 
\begin{equation}
 i \textstyle{\pddt}{\bF}^\ast(t,\sV)
= \label{eq:MLdotFinv}
       -  c\nab\times{\bF}^\ast(t,\sV) \, .
\end{equation}
 Notice that \refeq{eq:MLdotFinv} is identical with the negative of the complex conjugate of equation \refeq{eq:MLdotF}, and
therefore does not contain any new information.

 Thus, the real Maxwell system \refeq{eq:MLdotB}--\refeq{eq:MLdivB}, as field equations on Minkowski spacetime
which is covariant under the full Lorentz group, is equivalent to the pair of complex evolution equations 
\refeq{eq:MLdotF}, \refeq{eq:MLdotFinv} together with the constraint equation \refeq{eq:MLdivF} and its complex conjugate. 
 One can combine the pair of complex evolution equations into a single evolution equation,
and the pair of complex constraint equations into a single constraint equation, for a complex six component vector field
$\FV  = (\bF_+ ,\bF_-)^{\mathrm{T}}$ (cf. \cite{IBBphotonREV}); namely, we set $\nab\times\FV  := (\nab\times\bF_+ ,\nab\times\bF_-)^{\mathrm{T}}$ and 
$\nab\cdot\FV  := (\nab\cdot\bF_+ ,\nab\cdot\bF_-)^{\mathrm{T}}$, and define the block matrices
${\boldsymbol\Sigma}_1 = \begin{pmatrix} 0 & \mathds{1}_3 \\ \mathds{1}_3 & 0 \end{pmatrix}$ and 
${\boldsymbol\Sigma}_3 = \begin{pmatrix} \mathds{1}_3 & 0 \\ 0 & -\mathds{1}_3 \end{pmatrix}$ 
where $\mathds{1}_3$ and $0$ are the $3\times 3$ identity and zero matrices, then 
\begin{equation}
 i \textstyle{\pddt}{\FV}(t,\sV)
= \label{eq:MLdotFFstar}
         c {\boldsymbol\Sigma}_3 \nab\times{\FV}(t,\sV) \, 
\end{equation}
and
\begin{equation}
\textstyle
        {\boldsymbol\Sigma}_3 \nab\cdot {\FV}(t,\sV)  
= \label{eq:MLdivFFstar}
        0\, ,
\end{equation}
with $\FV$ satisfying the self-duality constraint
\beq 
\FV = {\boldsymbol\Sigma}_1 \FV^\ast.
\eeq
 This ``doubled up'' Weber form of Maxwell's equations is invariant under the full Lorentz group, yet
it does not contain any extra information --- it is strictly equivalent to Maxwell's field equations!

 Could $\FV$ (perhaps up to a conversion factor) be considered as ``wave function of the photon'', with the usual $L^2$
inner product Hilbert space?
 {As noticed by Rodi Tumulka,\footnote{Private communication, 2015.} 
since multiplication of a self-dual $\FV$ by a complex number does not preserve the self-duality, 
the space of all self-dual $\FV$ is not a complex Hilbert space, in violation of a basic quantum-mechanical principle.
 This alone should eliminate the self-dual $\FV$ from the list of contenders for the photon wave function --- 
unless for some subtle reason the photon Hilbert space should turn out not to be a complex Hilbert space!}\vspace{-5pt}

\subsection{Good's non-standard Hilbert space}\label{sec:GOODnogood}\vspace{-5pt}

  Invoking the conventional wisdom that the photon is its own anti-particle, Bia\l{}ynicki-Birula has proposed that $\FV$ 
(perhaps up to a conversion factor) is indeed to be considered as ``wave function of the photon'' \cite{BiBiTHREE,IBBphotonREV}.
 Of course, Bia\l{}ynicki-Birula is aware of the problem that 
 $|{\FV}|^2(t,\sV)$ cannot be the ``probability density of finding the photon at $\sV\in\Rset^3$, given $t\in\Rset$,'' 
for $|{\FV}|^2(t,\sV)$ does not transform like a probability density under Lorentz transformations.
 As we have seen already (since $\FV$ is equivalent to the pair $\bE\pm i\bB$),
it transforms like the time-time component of a rank-2 tensor, namely as an energy density (as was to be expected!). 

 In Bia\l{}ynicki-Birula's proposal \cite{BiBiTHREE,IBBphotonREV} (see also \cite{IBBZBBphotonREV}, \cite{RedkovETal})
$|\FV|^2$ is indentified with \emph{the density of the expected energy} ``of the photon,'' 
whereas his probability density w.r.t. Lebesgue measure for finding the photon at a point is
essentially equivalent to $|\psiLP|^2$, where 
   $\psiLP(t,\sV) \propto \int_{\Rset^3}\frac{\!\!\bF(t,\tilde{\sV})}{\;|\sV - \tilde{\sV}|^{5/2}}d^3s$
is the Landau--Peierls wave function \cite{LandauPeierls} --- a nonlocal object. 
 This would seem to suggest that $\psiLP$ is the quantum-mechanical photon wave function, except that it's not
--- Bia\l{}ynicki-Birula explains why $\psiLP$ is not acceptable as photon wave function.\footnote{The convolution of $\psiLP$ 
   with a ``polarization vector'' has been proposed as ``the photon's position wave function'' $f$ 
satisfying the ``massless square-root Klein--Gordon equation'' $i\hbar \pdt f = \hbar c \sqrt{-\Delta} f$; see \cite{Mostafazadeh}.}

 Interestingly enough, though, this unconventional identification can be fitted into a Hilbert space formulation (that apparently
goes back to \cite{Good1957}) if one is willing to accept a Hilbert space inner product not equivalent to the standard $L^2$ structure
but having a nonlocal kernel.
 More precisely, in order that (up to a factor $1/4\pi$) $\int_{\Rset^3}|\FV|^2d^3q$ is the (NB: now inevitably positive) expected energy 
``of the photon'' and \emph{not} the total probability [$=1$] for finding the photon, then in Dirac notation 
one needs to have $\int_{\Rset^3}|\FV|^2d^3q = \langle\FV|\hat{H}|\FV\rangle_{\!G}^{}$, 
with Hamiltonian $\hat{H} :=-i\hbar{\boldsymbol\Sigma}_3\widehat{\Si}\cdot\nab$, where $\widehat{\Si}$ 
and ${\boldsymbol\Sigma}_3$ were defined in Appendix A.
 But  spec $\hat{H} =\Rset$, so at best one has $\langle\FV|\hat{H}|\FV\rangle_{\!G}^{} = \pm \int_{\Rset^3}|\FV|^2d^3q$. 
 This is only possible if the inner product is given by\footnote{In \cite{Hawton} the inner product is denoted in honor 
   of Bia{\l}ynicki-Birula by $\langle\,.\,|\,.\,\rangle_{BB}^{}$;
   also, she presents a variation on the theme entirely in terms of the ``potentials'' $\AV$ for the ``fields'' $\FV$.} 
$\langle\FV_{(1)}|\FV_{(2)}\rangle_{\!G}^{} = \int_{\Rset^3} \FV_{(1)}^\dag |\hat{H}^{-1}|\FV_{(2)} d^3q$, 
where $\hat{H}^{-1}$ is the inverse of $\hat{H}$, and $|\hat{H}^{-1}|$ the absolute inverse Hamiltonian. 
 
 Even though the non-locality of this bilinear expression in the putative ``photon wave function'' $\FV$ means a 
significant departure from the usual formalism,\footnote{In textbook discussions of Dirac's equation for the electron, Born's rule is frequently 
stated as follows: ``$|\psi|^2(\qV)$ is  the probability density of finding the electron at $\qV$,'' suggesting the probability density 
is a \emph{pointwise} (local)  bilinear expression in the wave function.
 However, this glosses over the fact that $\psi$ in this statement is tacitly assumed normalized, which is a non-local operation.
 The proper statement of Born's rule, that ``the probability density of finding the electron at $\qV$ is given by 
 $|\psi|^2(\qV)/\int_{\Rset^3}|\psi|^2(\qV)\drm{q}^3$,'' makes it plain that a non-local operation on $\psi$ is involved.
 All the same, since $\int_{\Rset^3}|\psi|^2(\qV)\drm{q}^3$ is a constant of motion for Dirac's equation, at the end of the day
one only multiplies $|\psi|^2$ by a single number which is determined by the initial conditions; so the usual formalism operates with
a very mild --- and unavoidable --- form of non-locality: in the normalizing factor one integrates the product of the initial $\psi(\qV)$ 
with it's complex adjoint over all $\qV$. 
 By contrast, in the non-locality introduced by Good's inner product one multiplies $\psi(\qV)$ with the complex
adjoint of $\psi({\qV}^\prime)$, weighs this product with a weight factor which depends on the distance of $\qV$ to ${\qV}^\prime$, 
then one integrates this weighed product over all ${\qV}^\prime$, and finally over all $\qV$.\vspace{-10pt}}
but even more serious an objection is the \emph{structural instability} of this probability formula!
 Namely, the Hamiltonian depends sensitively on the ``environment'' --- as soon as the interactions of the photon 
with other particles (such as electrons) are taken into account. 
 And so the proposed formula for the probability density $|\psiLP|^2$ depends 
not just on the initial state (wave function) of the photon but also on the environment --- even at the initial instant.

 Incidentally, Bia\l{}ynicki-Birula proposes to vindicate Good's Hilbert space formalism by deriving it from the QED wave function 
formalism of an indefinite number of photons, normalized to $N=1$. 
 Recalling, as emphasized by Weinberg \cite{WeinbergTALK}, p.2, that ``second quantization'' does
not mean one quantizes a wave function --- what's being quantized is a ``classical field'' on physical spacetime ---, 
this vindication of Good's formalism from the QED wave function in momentum space representation (see Appendix D), 
normalized to $N=1$, is in effect but a reversal of the field quantization of the Maxwell equations for the electromagnetic fields in vacuum, 
sending one back to Maxwell's equations.
 Logically it does not vindicate $\FV$ as photon wave function, nor \refeq{eq:MLdotFFstar} as photon wave equation, with $\FV$ essentially
Weber's complex Maxwell field $\bE\pm i\bB$.

\subsection{The QED wave equation for an indefinite number of photons}

 All experts (e.g. \cite{OppiPHOTON,LandauLifshitzBOOKrelQM,WeinbergBOOKqft,IBBphotonREV,IBBZBBphotonREV, ScullyBOOK, Cohen-Tannoudji})
seem to agree that one can easily ``reverse-engineer'' a ``wave equation for a gas of photons''
with an \emph{indefinite} number of photons by applying a unitary transformation to the time evolution of the 
electromagnetic field operators of QED which act on ``the vacuum,'' switching from their ``Heisenberg picture'' 
to the ``Schr\"odinger picture.'' 

 For instance, the introductory textbook example is the wave function 
$\Phi(t) :=\{\Phi_{\kV,\lambda_\kV}(n;t)\,|\, \kV\in \Zset^3\!\setminus\!\{\boldsymbol{0}\},\lambda_\kV\in\{-1,1\}\}$ of  a 
``gas of non-interacting photons in a cubical box (of unit volume) with periodic boundary conditions,'' 
which is a $t$-dependent $(\kV,\lambda_\kV)$-indexed family of complex functions of the variable $n\in\{0\}\cup\Nset$,
satisfying the abstract Schr\"odinger equation
\beq \label{photonWEQUforINDEFINITELYmanyPHOTONS}
i\hbar \pdt \Phi(t) = H \Phi(t),
\eeq
where 
\beq \label{HforINDEFINITELYmanyPHOTONS}
H  = \sum_{\kV,\lambda_\kV} \hbar |\kV| ( N_{\kV,\lambda_\kV} + \textstyle\frac12) 
\eeq
is the Hamiltonian in occupation number representation, with $N_{\kV,\lambda_\kV}$ the number operator acting on
the component with wave-vector 
$\kV\in \Zset^3\!\setminus\!\{\boldsymbol{0}\}$ and polarization $\lambda_\kV\in\{-1,1\}$, 
--- i.e.,  
$\Phi_{\kV,\lambda_\kV}(n;t) =\sum_{n=0}^\infty \varphi^{n_{\kV,\lambda_\kV}}\Phi_{{\kV,\lambda_\kV}}^{n_{\kV,\lambda_\kV}}(n;t)$,
with $N_{\kV,\lambda_\kV}\Phi_{\kV,\lambda_\kV}^{n_{\kV,\lambda_\kV}}(n;t) = 
n_{\kV,\lambda_\kV}\Phi_{\kV,\lambda_\kV}^{n_{\kV,\lambda_\kV}}(n;t)$ being eigenfunctions of the number operator.

 Superficially this could be mistaken as also providing an example for a wave equation for a single photon because the different
modes don't interact. 
 However, in this approach the photon wave function for a single photon of wave vector $\kV$ and frequency $\omega = \pm c|\kV|$ is
simply the Fourier mode $e^{i(\kV\cdot\sV - \omega t)}$ (times a normalizing constant) restricted to the big box, and this
does not lead to an $L^2$ wave function when the box is expanded to  all of space. 
 Instead, a Fourier superposition is then needed, and since in this approach a ``photon'' is basically ``identified'' with a
wave vector $\kV$ and frequency $\omega = \pm c|\kV|$, it is manifest that such an $L^2(\Rset^3)$ mode inevitably consists of 
uncountably many ``photons'' --- if one insists that by ``photon'' one continues to mean an ``object with a 
wave vector $\kV$ and frequency $\omega = \pm c|\kV|$.''

 The upshot of all this is a moral: 
 It is important to keep in mind that the electromagnetic field is defined on physical spacetime $\Rset^{1,3}$, 
whether a point is occupied by a photon or not, whereas the quantum-mechanical wave function $\psiPH$ lives 
on the configuration space(-time) of a single photon's position variable, and so at each point represents some 
quantum-mechanical aspect of the photon when it is \emph{there}.
 It is simply \emph{conceptually} wrong to identify the quantum-mechanical wave function of the photon with the electromagnetic field,
or with some of its constituents, on physical space(-time).
 Likewise, the quantum-mechanical wave function of a single photon on its configuration space(-time) should also not be confused 
with the usual QED wave function, normalized to $N=1$.

\section{Spinors and bi-spinors: a primer}

\label{app:primer}
%
\subsection{Clifford algebras}\label{sec:clifford}
%
 Of central importance to the relativistic formulation of quantum mechanics in $d$ space dimensions is the 
{\em complexified spacetime algebra} $\cA$, defined as the complexification of the real Clifford algebra $\mbox{Cl}_{1,d}(\Rset)$ 
associated with the Minkowski quadratic form of signature $(+,-,\dots,-)$.  
 This is because scalar, vector, tensor and spinor-valued quantities defined on the spacetime $\Rset^{1,d}$, as well as operators 
and groups acting on them, can all be realized as $\cA$-valued sections over the spacetime, allowing for an efficient formulation 
of the theory  (see e.g. \cite{Hes2015}.) 
 In three space dimensions, $\cA = \mbox{Cl}_{1,3}(\Rset)_\Cset$ is isomorphic to $M_4(\Cset)$, the algebra of $4\times 4$ matrices 
with complex entries. 
 The isomorphism can be realized by choosing a basis for $\mbox{Cl}_{1,3}(\Rset)$ and taking complex linear  combinations of the basis elements. 
 A convenient basis for the real algebra is formed  by the so-called Dirac gamma-matrices (in their Weyl representation) and their products: 
 Let $\mathds{1}_n$ denote the $n\times n$ identity matrix, and define
\beq\label{def:gammas}
\ga^0 = \left( \begin{array}{cc}0 & \mathds{1}_2 \\ 
\mathds{1}_2 & 0 \end{array} \right),\qquad \ga^k = \left( \begin{array}{cc}0 & -\si_k \\
 \si_k & 0 \end{array} \right),\ k=1,2,3.
\eeq
where the $\si_{k\in\{1,2,3\}}$ are the three conventional Pauli matrices. 
 The $\ga$-matrices satisfy the Clifford algebra relations
\beq \label{gammasON}
\ga^\mu \ga^\nu + \ga^\nu \ga^\mu = 2 \eta^{\mu\nu}\Id_4;
\eeq
where the $\eta^{\mu\nu}$ are the components of the metric tensor 
\beq
\boldsymbol{\eta} = \diag(1,-1,-1,-1).
\eeq 

Any Clifford algebra $A$ associated with a vector space $V$ over a field $F$ contains a subspace that is isomorphic to $V$. 
 The elements of that subspace are called {\em 1-vectors}.  
Note that \refeq{gammasON} simply states that the four matrices $\{\ga^\mu\}_{\mu=0}^3$ form a Lorentz-orthonormal set of 
1-vectors in the Clifford algebra $\mbox{Cl}_{1,3}(\Rset)$.  

The following frequently-used embedding of $\Rset^{1,3}$ into its Clifford algebra $\mbox{Cl}_{1,3}(\Rset)$ can therefore be 
thought of as the expansion of a 1-vector in an ortho-normal basis:  For $x = (x^\mu)_{\mu=0}^3 \in \Rset^{1,3}$ let 
\beq\label{embed}
\ga(x) := \ga_\mu x^\mu.
\eeq 
(Indices are raised and lowered using the Minkowski metric $\boldsymbol{\eta}$ and repeated indices are summed over the range 0 to 3.)
 We will see in the next subsection that $\ga(x)$ is a rank-two bi-spinor.  

 By definition, a {\em $k$-vector} is the (Clifford) product of $k$ elements, each one of which is a 1-vector. 
 Let $\{\be_j\}_{j=1}^n$ be a basis for $V$. Every Clifford number $a\in A$ has a $k$-vector expansion of the form 
\beq a = a_S\Id + \sum_{k=1}^{n}\sum_{1\leq i_1<\dots<i_k\leq n} a^{i_1\dots i_k} \be_{i_1}\dots\be_{i_k},
\eeq
where $n=\dim_F V$ and the coefficients $a_S,a^{i_1\dots i_k}$ are in $F$.  
 It follows that the following is a basis for the (16-dimensional) algebra $\mbox{Cl}_{1,3}(\Rset)$:
\beq\label{basis}
\cB := \left\{ \mathds{1}_4; \ga^0, \ga^1, \ga^2,\ga^3; \ga^0\ga^1, \dots; \ga^0\ga^1\ga^2, \dots;\dots;\ga^0\ga^1\ga^2\ga^3 \right\},
\eeq
and therefore its complexification can be obtained by taking the coefficients of the expansion to be complex numbers: $\cA = \Span_\Cset \cB$.  

 The complexified algebra $\cA$ in particular includes the {\em pseudoscalar}
\beq 
\ga^5 := i \ga^0 \ga^1 \ga^2 \ga^3 
= \begin{pmatrix} \, \Id_2 & \ 0 \\ 0 & -\Id_2\end{pmatrix},
\eeq
and therefore the projections
\beq\label{def:projections}
\Pi_\pm := \half( \Id_4 \pm \ga^5).
\eeq
Using these projections it follows right away that $\cA$ contains all $4\times 4$ matrices, 
and it is easy to verify that $\cA$ and $M_4(\Cset)$ are indeed isomorphic as algebras, with the 
Clifford multiplication given by matrix multiplication.

 Let $a = a_S \Id +\sum_I a_I \bga^I$ denote the $k$-vector expansion of $a\in \cA$. 
 Thus $a_S,a_I\in \Cset$ and each $\bga^I$ is a $k$-fold product of gamma matrices, for some $k$. 
 Two important operations on Clifford numbers are the following:  
\begin{itemize}
\item 
The {\em scalar part} $a_S$ of an element $a \in \cA$ is by definition the coefficient of the unit element $\Id$ in the expansion of $a$ in 
any basis (such as $\cB$.)  
Using the isomorphism above, we can view $a$ as a $4\times 4$ matrix, and we then have
\beq 
a_S = \frac{1}{4}\tr a,
\eeq
where $\tr$ denote the usual operation of taking the trace of a matrix.
\item 
The {\em conjugate reversion} (a.k.a. {\em Dirac adjoint}) $\overline{a}$ of  $a\in\cA$ is by definition the element 
obtained by reversing the order of multiplication of the 1-vectors in the expansion of $a$ in terms of $k$-vectors, 
and taking the complex conjugate of the coefficients in that expansion. 
 Thus
$
\overline{a} = a_S^* \Id + \sum_I a_I^* \tilde{\bga}^I 
$
with $\tilde{\bga}^I = \ga^{i_k}\dots\ga^{i_1}$ whenever $\bga^I = \ga^{i_1}\dots\ga^{i_k}$. 
 Using the isomorphism $\cA \cong M_4(\Cset)$ it is not hard to see that, when $a$ is viewed as a $4\times 4$ matrix,
\beq\label{def:Diradj}
\overline{a} = \ga^0 a^\dag \ga^0
\eeq
where $a^\dag = (a^*)^{\mathrm{T}}$ denotes the conjugate-transpose of $a$. 
(Here and elsewhere, ${}^\ast$ denotes complex conjugation, while ${}^{\mathrm{T}}$ denotes the matrix transpose.)
\end{itemize}
\begin{rem}\label{rem:diradj}
\textit{Despite the appearance of the $\ga^0$ factors in \refeq{def:Diradj}, the Dirac adjoint operation $a\mapsto \overline{a}$ 
is frame-independent.
  Indeed, on Lorentzian backgrounds, it is the dagger operation $a \mapsto a^\dag$ that requires the choice of a frame (more precisely,
 that of a time-like direction) to be made. 
 For a quick proof of this and other important facts about Clifford numbers see \cite{Rie1946}.}
\end{rem}

\subsection{A brief recap of the groups $SL(2;\Cset)$, $SO_0(1,3;\Rset)$, $O(1,3;\Rset)$}
\label{sec:LorentzRep}
 Consider the mapping $\si :\Rset^4 \to \Cset^{2\times 2}$ given by
\beq \label{def:si}
\si(\bx) := x^0 \si_0 + x^1\si_1 + x^2\si_2 + x^3 \si_3 
\eeq
with $\si_0 := \mathds{1}_2$ (formally:
$\si(\bx) = x^\mu \si_\mu$).
  It is known that for real $\bx$, $\si(\bx)$ is a Hermitian matrix, and in fact $\{\si_\mu\}_{\mu=0}^3$ is a basis for the set of
Hermitian matrices, $H(2)$, as a vector space over $\Rset$, so that any Hermitian matrix H $\in H(2)$ can be written in the form
H $=\sigma(\bx)$ for a vector $\bx\in\Rset^4$.  

  On the other hand, if $A \in SL(2,\Cset)$ is any complex-valued $2\times 2$ matrix of determinant one, and $\bx \in \Rset^4$, then
 clearly $A\si(\bx)A^\dagger$ is Hermitian.
 It follows (see \cite{ThallerBOOK} Ch. 2 for details) that $\si$ provides a homomorphism between $SL(2,C)$ and the {\em proper} 
Lorentz group  $SO_0(1,3)\equiv\mathcal{L}^\uparrow_+$, i.e.
\beq \label{six}
A \si(\bx) A^\dagger = \si(\by) \iff \by = \bLa_A \bx,
\eeq
where $\bLa_A$ satisfies $\bLa_A^{\mathrm{T}} \boldsymbol{\eta} \bLa_A = \boldsymbol{\eta}$. 
 Thus $A \mapsto \bLa_A$ is a group homomorphism and a covering map.
 In fact,  $SL(2,\Cset)$ is the universal cover of $SO_0(1,3)$.
 This is a 2-to-1 map, because $\bLa_{\pm \Id_2 }= \Id_4$.

 Another, inequivalent, representation of the proper Lorentz group can be obtained using 
\beq \label{def:sipx}
\si'(\bx) := x^0 \si_0 - x^1\si_1 - x^2\si_2 - x^3 \si_3.
\eeq
 Setting $A^{-\dag} := {A^{-1}}^\dag \equiv {A^\dag}^{-1}$, one easily checks that 
\beq \label{siprx}
A^{-\dag} \si'(\bx) A^{-1} = \si'(\by) \iff \by = \bLa_A \bx.
\eeq

 One can use these two inequivalent representations of the proper Lorentz group $SO_0(1,3)$  to create a representation of the full 
Lorentz group $O(1,3)$ (see \cite{ThallerBOOK},  {\bf 2.5}). 
 One first doubles up the representation of the proper Lorentz group obtained above by defining 
\beq 
\bL_A := \left(\begin{array}{cc} A & 0 \\ 0 & A^{-\dagger} \end{array}\right),
\eeq
which still covers the proper Lorentz group.
  One may then complement the $\bL_A$ by the matrices 
\beq 
\bL_P = \left(\begin{array}{cc} 0 & \Id_2\\ \Id_2 & 0 \end{array}\right) = \ga^0,
\qquad\mbox{and}\qquad
 \bL_T = \left(\begin{array}{cc} 0 & -i\Id_2\\ i\Id_2 & 0 \end{array}\right),
\eeq
representing the improper Lorentz transformations \emph{space inversion} 
\beq\label{def:Pinver}
P := \left(\begin{array}{cc} 1 & 0^{\mathrm{T}} \\ 0 & -\Id_3\end{array}\right),
\eeq
(a.k.a. parity transformation) 
and \emph{time reversal} $T := -P$.
 Then $\{ \bL_A \ |\ A \in SL(2,\Cset)\}\cup\{\bL_P\}\cup\{\bL_T\}$ generates the universal covering group $\fL$ 
of the full Lorentz group.\footnote{Since the full Lorentz group is not connected, it has no unique covering group. 
Indeed, there are eight non-isomorphic covering groups of the full Lorentz group, leading to eight inequivalent representations, 
the above being only one of them.  
 However, all of these yield the same {\em projective} representation of the full group, and 
therefore are equivalent from the point of view of quantum mechanics. 
 Under certain assumptions, it is possible to reduce the number of admissible representations of the covering group to four, 
corresponding to those for which the time-reversal operator is represented by an {\em anti}-unitary map. 
See \cite{ThallerBOOK} pp.76 and 104, in particular his Thm.~3.10.}  

 N.B.: It is a simple consequence of the above definitions that $\forall\ \bL \in \fL$ we have
\beq \label{LPprop}
\bL^\dag \bL_P \bL = \bL_P.
\eeq

\subsection{Rank-one spinors}

 A rank-one spinor field is a section of a fiber bundle over the Minkowski space $\Rset^{1,3}$, with fibers isomorphic to $\Cset^2$, 
that transforms in a particular way under the action of the proper Lorentz group on the base.
  There are two types of (contra-variant) rank-one spinors: those with one undotted index (said to be of type $(\half,0)$), and those 
with one dotted index (said to be of type $(0,\half)$).
 The rank-one spinor $\zeta := (\zeta^a)_{a=0,1}$ of type $(\half,0)$ transforms as $\zeta \mapsto A\zeta$ 
under a base transformation $\bx \mapsto \bLa_A \bx$ with $\bLa_A \in SO_0(1,3)$ for some $A\in SL(2,\Cset)$.
 Similarly, $\dot\zeta := (\zeta^{\dot{a}})_{\dot{a}=0,1}$, a rank-one spinor field of type $(0,\half)$ 
transforms as $\dot\zeta \mapsto A^{-\dagger}\dot\zeta$ under the action of $\bLa_A$ as above.
 Spinor indices, whether dotted or undotted, can be moved up and down using the Levi-Civita symbols 
$(\ep^{ab}) = \left(\begin{array}{cc} 0 & -1\\ 1 & \ 0\end{array}\right)$ and $(\ep_{ab}) = -(\ep^{ab})$.
 For the remainder of this primer, when we refer to rank-one spinors we mean those of the contra-variant kind, i.e. with one index upstairs.

\subsection{Rank-two spinors}

 Rank-two spinors map rank-one spinors into rank-one spinors. 
 A rank-two spinor thus has two spinorial indices and comes in four varieties, since each of its indices could be either dotted or 
undotted. 
 These four types are distinguished by the way they transform under a proper Lorentz transform of the base: 
\beq 
\varrho^a_b 
 \mapsto  A^a_c \varrho^c_d A^{-1}{}^d_b, \quad
\upsilon^a_{\dot{b}} \mapsto  A^a_c \upsilon^c_{\dot{d}} A^\dag{}^{\dot{d}}_{\dot{b}}, \quad
\vartheta^{\dot{a}}_b \mapsto  A^{-\dag}{}^{\dot{a}}_{\dot{c}} \vartheta^{\dot{c}}_d A^{-1}{}^d_b, \quad
\varsigma^{\dot{a}}_{\dot{b}}  
\mapsto  A^{-\dag}{}^{\dot{a}}_{\dot{c}} \varsigma^{\dot{c}}_{\dot{d}} A^\dag{}^{\dot{d}}_{\dot{b}}.
\eeq
 More concisely,
\beq \label{transrules}
\varrho  \mapsto  A \varrho A^{-1}, \quad
\upsilon \mapsto  A \upsilon A^\dag{}, \quad
\vartheta \mapsto  A^{-\dag}\vartheta A^{-1}, \quad
\varsigma \mapsto  A^{-\dag}\varsigma A^\dag{}.
\eeq
 It follows from \refeq{six} and \refeq{siprx} that for $\bx\in\Rset^{1,3}$, $\si(\bx)$ and $\si'(\bx)$ are Hermitian rank-two spinors,
which transform as a rank-two spinor of the type $\upsilon^a_{\dot{b}}$, respectively 
of the type $\vartheta^{\dot{a}}_b$. 

 Now a general $2\times 2$ matrix with complex entries can always be written as the sum of a Hermitian and an anti-Hermitian matrix. 
 It thus follows that an arbitrary rank-two spinor of type $(\upsilon^a_{\dot{b}}) = \si(\bz)$ for some $\bz \in \Cset^4$, and similarly 
$(\vartheta^{\dot{a}}_b) = \si\rq{}(\bz)$.

 Just as the mappings $\bx\mapsto \si(\bx)$ and $\bx\mapsto \si'(\bx)$ provide a correspondence between \emph{real four-vectors} and 
\emph{Hermitian rank-two spinors}, there are also mappings that provide an identification of 
\emph{real anti-symmetric four-tensors of rank two} with \emph{traceless rank-two spinors} of types $\varrho^a_b$ and
$\varsigma^{\dot{a}}_{\dot{b}}$: 
  Let  $\bff = (\ff_{\mu\nu})$ be a real anti-symmetric four-tensor of rank two (a two-form on Minkowski spacetime, 
in other words) defined on the configuration spacetime $\Rset^{1,3}$.
  Let $\star$ denote the Hodge star operator. 
We define a {\em dual} ${}^{\mathrm{d}}\bff$ of $\bff$ to be the two form 
\beq 
{}^{\mathrm{d}}\bff :=  i \star\bff,\quad{\mbox{i.e.}}\quad {}^{\mathrm{d}}\ff_{\mu\nu} = \frac{i}{2} \ep_{\alpha \beta \mu\nu} \ff^{\alpha\beta},
\eeq
where the $\ep_{\alpha \beta \mu\nu}\in\{-1,0,1\}$ are the coefficients (w.r.t. a Lorentz frame; see below) 
of the totally anti-symmetric four-tensor given by the (hyper)volume form of Minkowski spacetime.
 The tensor $\bff$ can be decomposed into a {\em self-dual} and an {\em anti-self-dual part} as follows (cf. \cite{Pen1976}): Define
\beq 
{}^{\mathrm{sd}}\bff := \bff + i \star\bff,\qquad {}^{\mathrm{asd}}\bff := \bff - i \star\bff
\eeq
so that $\bff = (\sdf + \asdf)/2$.
 Since $\star\star{\bff}=-\bff$, it follows that  ${}^{\mathrm{d}}\left(\sdf\right) = \sdf$ and ${}^{\mathrm{d}}\left(\asdf\right) = - \asdf $.

 Let us now define (cf. \cite{TamburiniVicino}) the two mappings $\Si$ and $\Si'$, taking a 2-form to a rank-two spinor as follows:
\beq \label{def:Siff}
\Si(\bff) := \frac{i}{4} \ff^{\mu\nu} \si_\mu \si'_\nu,\qquad 
\Si'(\bff) := \frac{i}{4} \ff^{\mu\nu} \si'_\mu \si_\nu.
\eeq
 It is easy to see that $\Si(\bff)$ is a rank-two spinor with two undotted indices, and $\Si'(\bff)$ is a rank-two spinor with two dotted 
indices, since they conform to the transformation rules for $\varrho$ and $\varsigma$ in \refeq{transrules}.
 Furthermore, the trace of these spinors (as linear transformations of rank-one spinors) vanishes\footnote{Incidentally, 
as the astute reader may have already surmised, the trace condition \refeq{defpsiPH} on our photon wave function can be 
justified by the fact that a rank-two bi-spinor with {\em pure trace} diagonal blocks, i.e. $\psi_\pm = u_\pm \si_0$ with 
$u_\pm \in \Cset$, would have something to do with a spin-zero object. We will pick up on this thread in a future publication.}.

 One can also check that $\Si(\asdf) = 0$ and $\Si'(\sdf) = 0$, i.e. $\Si(\bff)$ only depends on the self-dual part of $\bff$ and 
$\Si'(\bff)$ depends only on the anti-self-dual part.  
It is also a consequence of the above definitions and properties that 
\beq 
\left( \Si(\bff)\right)^\dag = \Si'(\bff).
\eeq
 Finally, it is easy to 
check that given any trace-free rank-two spinor with components $\varrho^a_b$ (resp. $\varsigma^{\dot{a}}_{\dot{b}}$) there 
exists a real anti-symmetric four-tensor of rank two, $\bff$, such that $\varrho = \Si(\bff)$ (resp. $\varsigma = \Si'(\bff)$). 
 This establishes the correspondence.

 Even though it is not necessary for the above assertions, for computational purposes it is often advantageous to use a coordinate 
frame and to work with the components relative to that frame. 
 To that end we may choose a particular Lorentz-orthonormal frame $\{\bu_{(\mu)}\}_{\mu = 0}^3$ in $\Rset^{1,3}$ and express the tensor 
$\bff$ in this frame. 
 Let us define the real three-vectors $\be,\bb \in \Rset^3$ by their components,
\beq \label{defeb}
f^{0k} =: -e^k,\qquad \star{f}^{0k} =: -b^k,\quad k=1,2,3, 
\eeq
where the indices refer to the components in the chosen frame, e.g. $f_{\mu\nu} = \bff(\bu_{(\mu)},\bu_{(\nu)})$, and where
$f^{\mu\nu} = \eta^{\mu\al}\eta^{\nu\beta} f_{\al\beta}$; and let us define also the complex three-vector
\beq  
\beff := \be + i \bb, 
\eeq
and the complex four-vector
\beq 
\xiV = \begin{pmatrix} 0 \\  \beff \end{pmatrix},
\eeq
where the components at r.h.s. are w.r.t. the frame $\{\bu_{(\mu)}\}_{\mu = 0}^3$.
 Then one readily checks that
\beq 
\Si(\bff) = i\xi^\al\si_\al = i\si(\xiV),\qquad \Si'(\bff) = -i\xi^\ast{}^\al \si_\al=i\si'(\xiV^\ast).
\eeq
 Note that this correspondence, unlike the ones we have introduced in the above, is frame-dependent, and will not be preserved under a 
Lorentz transformation.

\subsection{Rank-one bi-spinors}

 As mentioned before, in order to obtain a representation of the full Lorentz group, one that includes parity transformations, 
one needs to double up the dimension of the representation and go to bi-spinors.  
A rank-one bi-spinor field $\psi$ is a section of a  bundle with $\Cset^4$ fibers, such that the top two components transform as a 
rank-one spinor of type $(\half,0)$ while the bottom two components transforms as a rank-one spinor of type $(0,\half)$:
\beq 
\psi = \left( \begin{array}{c} (\zeta_-{}^a) \\ (\zeta_+{}^{\dot{a}}) \end{array}\right).
\eeq
 Dirac's wave function of the electron is of this type. 
 Under the action of an element $\bLa$ of the full Lorentz group, a rank-one bi-spinor transforms like
\beq  
\psi \mapsto \bL \psi ,
\eeq
where $\bL \in \fL$ is the projective representation of $\bLa$.  

\subsection{Rank-two bi-spinors}
\label{sec:defPWF}
 Finally, by analogy with rank-two spinors, rank-two bi-spinors (also known as {\em Clifford numbers} \cite{Rie1946}) are operators acting on rank-one bi-spinors, and producing a rank-one 
bi-spinor. 
 These can be represented by a $4\times 4$ complex-valued matrix with four $2\times 2$ blocks, each one of which is a rank-two spinor 
of a certain type:
\beq \label{rank2bisp}
\psi = \left( \begin{array}{cc} \varrho & \upsilon \\
\,\vartheta &\, \varsigma \end{array} \right)
\equiv \left( \begin{array}{cc} (\varrho^a_b) & (\upsilon^a_{^{\dot{b}}}) \\
\,(\vartheta^{\dot{a}}_b) &\, (\varsigma^{\dot{a}}_{\dot{b}} )\end{array} \right).
\eeq
 The rules of transformation of these four rank-two spinors imply that a rank-two bi-spinor under the action of the Lorentz group would 
transform as
\beq 
\psi \mapsto \bL \psi \bL^{-1},
\eeq
where $\bL \in \fL$ is the projective representation of $\bLa$.  


\baselineskip=10pt
\bibliographystyle{plain}

\end{document}